\documentclass[ukenglish,notitlepage,a4paper,11pt]{article}
\usepackage{natbib}
\usepackage{array}
\usepackage{multirow}
\usepackage{amsthm}
\usepackage{amsmath}
\usepackage{amsfonts}
\usepackage{amssymb}
\usepackage{pgf}
\usepackage{bm}
\usepackage{paralist}
\usepackage{pdflscape}
\usepackage{rotating}
\usepackage{scalefnt}
\usepackage{xr}
\usepackage{verbatim}
\usepackage[flushleft]{threeparttable}
\usepackage[bottom]{footmisc}
\usepackage{amsmath}
\usepackage{graphicx}
\usepackage{epstopdf}
\usepackage{setspace}
\usepackage{booktabs}
\usepackage{multirow}
\usepackage{ccaption}
\usepackage{amsfonts}
\usepackage{url}
\usepackage{rotating}
\usepackage{colortbl}
\usepackage{lscape}
\usepackage{xcolor}
\usepackage{amssymb}
\usepackage{float}
\usepackage[font=small,labelfont=bf]{caption}
\usepackage{titlesec}
\usepackage{subcaption}
\usepackage{amsmath}
\usepackage{makeidx}
\usepackage{scalefnt}
\usepackage{epstopdf}
\usepackage{epsfig}
\usepackage{sectsty}
\usepackage[flushleft]{threeparttable}
\usepackage{array}
\usepackage{arydshln}
\usepackage{natbib}
\usepackage{titling}
\usepackage{xpatch}
\usepackage{amsthm}
\usepackage[normalem]{ulem}
\usepackage{array}
\usepackage{natbib}
\usepackage{geometry}
\usepackage{graphicx}
\usepackage{array}
\usepackage{lmodern}
\usepackage{multirow}
\usepackage{amsthm}
\usepackage{amsmath}
\usepackage{amsfonts}
\usepackage{amssymb}
\usepackage{adjustbox}
\usepackage{pgf}
\usepackage{bm}
\usepackage{paralist}
\usepackage{float}
\usepackage{pdflscape}
\usepackage{rotating}
\usepackage{scalefnt}
\usepackage{xr}
\usepackage{verbatim}
\usepackage[flushleft]{threeparttable}
\usepackage{graphicx}
\usepackage{chngpage}
\usepackage[skip=20pt]{caption}
\usepackage{enumitem}
\usepackage{caption}
\usepackage{subcaption}
\usepackage{lmodern,blindtext}
\usepackage{indentfirst}
\usepackage[colorlinks,urlcolor = black,pdfpagelabels,pdfstartview = FitH,
  bookmarksopen = true,bookmarksnumbered = true,linkcolor = black,
  plainpages = false,hypertexnames = true,citecolor = blue,
]{hyperref}

\newtheorem{theorem}{Theorem}

\newtheorem{lemma}{Lemma}

\newtheorem{assumpB}{Assumption}
\newtheorem{assumpC}{Assumption}

\renewcommand{\theequation} {\arabic{section}.\arabic{equation}}

\newcommand{\ignore}[1]{}

\newcommand{\beq}{\begin{equation}}
\newcommand{\eeq}{\end{equation}}
\newcommand{\bea}{\begin{eqnarray}}
\newcommand{\eea}{\end{eqnarray}}
\newcommand{\bit}{\begin{itemize}}
\newcommand{\eit}{\end{itemize}}
\newcommand{\ben}{\begin{enumerate}}
\newcommand{\een}{\end{enumerate}}
\newcommand{\bpm}{\begin{pmatrix}}
\newcommand{\epm}{\end{pmatrix}}
\newcommand{\bbm}{\begin{bmatrix}}
\newcommand{\ebm}{\end{bmatrix}}

\titleformat{\section}{\Large\bfseries}{\thesection}{1em}{}

\font\myfont=cmr12 at 16pt

\date{}
\begin{document}

\title{\textbf{Online Change-point Detection for Matrix-valued Time Series
with Latent Two-way Factor Structure}}
\author{ \myfont Yong He\thanks{%
Shandong University, Email:heyong@sdu.edu.cn}, Xin-Bing Kong\thanks{%
Nanjing Audit University, Email:xinbingkong@126.com}, Lorenzo Trapani
\thanks{%
University of Nottingham, Email:Lorenzo.Trapani@nottingham.ac.uk}, Long Yu%
\thanks{%
National University of Singapore, Email:stayl@nus.edu.sg} }
\maketitle

\begin{center}
\begin{minipage}{125mm}
			
			
			{\it Abstract}: This paper proposes a novel methodology for the online detection of changepoints in the factor structure of large matrix time series. Our approach is based on the well-known fact that, in the presence of a changepoint, a factor model can be rewritten as a model with a larger number of common factors. In turn, this entails that, in the presence of a changepoint, the number of spiked eigenvalues in the second moment matrix of the data increases. Based on this, we propose two families of procedures - one based on the fluctuations of partial sums, and one based on extreme value theory - to monitor whether the first non-spiked eigenvalue diverges after a point in time in the monitoring horizon, thereby indicating the presence of a changepoint. Our procedure is based only on rates; at each point in time, we randomise the estimated eigenvalue, thus obtaining a normally distributed sequence which is \textit{i.i.d.} with mean zero under the null of no break, whereas it diverges to positive infinity in the presence of a changepoint. We base our monitoring procedures on such sequence. Extensive simulation studies and empirical analysis justify the theory.

			\bigskip
			
			
		\end{minipage}
\end{center}



\section{Introduction}

Matrix-valued time series have been paid substantial attention in such
diverse areas as finance, biology and signal processing. In this work, we
focus on developing several on-line changepoint detection schemes for
matrix-valued time series with an (approximate) latent factor structure. To
fix ideas, let $\left\{ X_{t},t=1,2,...\right\} $ be a time series of $%
p_{1}\times p_{2}$ matrices. A matrix factor model with common factors%
\textbf{\ }can be written as%
\begin{equation}
\left( X_{t}\right) _{p_{1}\times p_{2}}=\left( R\right) _{p_{1}\times
k_{1}}\left( F_{t}\right) _{k_{1}\times k_{2}}\left( C^{\prime }\right)
_{k_{2}\times p_{2}}+\left( E_{t}\right) _{p_{1}\times p_{2}}, \ k_1, \
k_2>0,  \label{fm}
\end{equation}%
where the subscripts represent the row and column dimensions of each matrix.
In (\ref{fm}), $R$ is a $p_{1}\times k_{1}$ matrix of loadings explaining
the variations of $X_{t}$ across the rows, $C$ is a $p_{2}\times k_{2}$
matrix of loadings reflecting the differences across the columns of $X_{t}$
, $F_{t}$ is a matrix of common factors, and $E_{t}$ is the idiosyncratic
component, and we assume that factor numbers $k_1$ and $k_2$ are positive,
demonstrating a collaborative dependence between both the cross-row and the
cross-column dimensions. The matrix factor structure in (\ref{fm}) is also
known as a two-way factor structure, and we use the same convention as in
\cite{hkyt} that
\begin{equation}
\left( X_{t}\right) _{p_{1}\times p_{2}}=\left\{
\begin{array}{ll}
\left( R\right) _{p_{1}\times k_{1}}\left( F_{t}\right) _{k_{1}\times
p_{2}}+\left( E_{t}\right) _{p_{1}\times p_{2}}, & k_{1}>0,k_{2}=0, \\
\left( F_{t}\right) _{p_{1}\times k_{2}}(C^{\prime })_{k_{2}\times
p_{2}}+\left( E_{t}\right) _{p_{1}\times p_{2}}, & k_{2}>0,k_{1}=0, \\
\left( E_{t}\right) _{p_{1}\times p_{2}}, & k_{1}=k_{2}=0,%
\end{array}%
\right.  \label{fm1}
\end{equation}%
to indicate that there may be no column common factors (first case), no row
common factors (second case) or no factor structure at all (third case). The
model without either column or row common factors (corresponding to the
first and the second case) is known as one-way/vector factor structure in
the literature. A possible approach to analyse a two-way factor model like (%
\ref{fm}) is to first vectorize the data $X_{t}$, and then to employ the
techniques which have been developed for the classical vector factor models.
However, when data genuinely have a two-way factor structure as in (\ref{fm}%
), this approach is bound to lead to sub-optimal inference (see, for
example, \citealp{wang2019factor}, and \citealp{Yu2021Projected}):
\textquotedblleft \lbrack ...] analyzing large scale matrix-variate data is
still in its infancy, and as a result, scientists frequently analyse
matrix-variate observations by separately modelling each dimension or
`flattening' them into vectors. This destroys the intrinsic
multi-dimensional structure and misses important patterns in such large
scale data with complex structures, and thus leads to sub-optimal
results\textquotedblright , as quoted from \citet{fan2021}.

\subsection{Literature review}

Given the importance of developing an inferential theory specifically for a
model like (\ref{fm}), in recent years, factor models for matrix-valued time
series have been paid significant attention as an alternative to vectorising
the data $X_{t}$ and applying a standard factor model. Indeed, many
contributions study several aspects of (\ref{fm}). \citet{wang2019factor} -
in a similar spirit to \citet{LY12} - propose an estimator of the factor
loading matrices (and of numbers of the row and column factors) based on an
eigen-analysis of the auto-cross-covariance matrix, under the assumption
that the idiosyncratic term $E_{t}$ is white noise. Assuming cross-sectional
pervasiveness along the row and column dimensions, \citet{fan2021} propose
an estimation technique for (\ref{fm}) based on an eigen-analysis of a
weighted average of the mean and the column (row) covariance matrix of the
data; \citet{Yu2021Projected} improve the estimation efficiency of the
factor loading matrices with iterative projection algorithms (see also %
\citealp{hkyt}). Extensions of the basic set-up in (\ref{fm}) include the
constrained version by \citet{chen2019constrained}, the semiparametric
estimators by \citet{chen2020semiparametric}, and the estimators developed
in \citet{chen2021factor}; see also \citet{han2020rank}. Applications of (%
\ref{fm}) include the dynamic transport network in the context of
international trade flows by \citet{Chen2020Modeling} and the financial data
analysis by \citet{chen2020testing}.

Conversely, other aspects in the inference on (\ref{fm}) are yet to be
explored. To the best of our knowledge, there is no contribution on
determining whether the factor structure in (\ref{fm}) is constant over time
or not (one exception is the paper by \citealp{chenthreshold}, who study the
in-sample inference in the presence of a threshold structure). This is a
crucial question in many applications: among the many possible examples, in
marketing science (where matrix factor models arise e.g. in the analysis of
movie recommendation rating matrix, see the seminal paper by \citealp{koren}%
), finding evidence of changepoints would help to learn updated consumer
preferences; in asset pricing, the presence of a changepoint in series of
matrices of several financials recorded for several companies could
highlight the presence of e.g. a downside state of the market; and in
climate studies, detecting changes in an array of environmental variables
observed at a group of stations could prove very useful in monitoring
pollution; in the context of physiological time series, matrix factor models
can be naturally applied to the analysis of EEG data (where signals are
recorded for several patients and several electrodes), and the presence of a
changepoint is a useful marker to spot an epileptic manifestation (%
\citealp{lavielle}). In this paper, we are particularly interested in
\textit{online detection} of a changepoint, i.e. in the timely detection of
a possible break as new data come in. Finding a changepoint could reveal the
heterogeneity of the matrix generation mechanism in machine learning, help
to explain economic events in economics, learning updated consumer
preferences in marketing science, and so on. Further, the timely detection
of a changepoint has also profound implications on model selection, helping
the applied user to amend her/his forecasts as soon as the change is
detected.

Changepoint detection in a high-dimensional context has received lots of
attention recently, both in terms of offline, in-sample and online,
out-of-sample, detection - see e.g. the recent contributions by %
\citet{samworth1} and \citet{samworth3} (see also \citealp{samworth2}, for
offline detection), where a change in the mean in a high-dimensional vector
is tested for. As far as detection of changes in the second-order structure
of high dimensional vectors is concerned, in-sample changepoint detection is
well-developed in the context of vector factor models. Recently, the
literature has proposed a series of tests for the in-sample detection of
breaks in factor structures: examples include the works by %
\citet{breitung2011}, \citet{chen2014}, \citet{han2014}, %
\citet{corradi2014testing}, \citet{yamamoto2015}, \citet{cheng2016}, %
\citet{BCF16}, \citet{baltagi2015}, and \citet{baltagi2021} (see also %
\citealp{bkf}). Although, as mentioned above, all these contributions
consider in-sample detection of changepoints, the techniques developed
therein could in principle be adapted/extended to the on-line detection
problem. Furthermore, \citet{BT1} develop a procedure aimed at sequential
changepoint detection.

Hence, a possible approach to detect changes in a matrix factor model could
be based on vectorising the data $X_{t}$, and then applying the techniques
which have been developed for the changepoint analysis of vector factor
models. However, such an approach would be subject to the criticisms
mentioned above, since it ignores the spatial relationship of an array - for
example, in the above example of movie recommendation rating matrices,
customers with certain common characteristics tend to favour some featured
movies, and this would be missed by a vector factor model.

\subsection{ Contributions and structure of the paper}

In this paper, we develop a procedure for the online detection of
changepoints in a matrix factor model; in particular, we study the detection
of changes in the row and column factor spaces, spanned by the columns of $R$
and $C$, due e.g. to the switching, enlargement or contraction of the
spaces. As is typical in the literature, we use two well-known facts. First,
in a model with a common factor structure, the second moment matrix of the
data has a spiked spectrum: in the presence of, say, $k$ common factors, the
first $k$ eigenvalues diverge with the sample size, whereas the others are
bounded. Second, as noted in \citet{corradi2014testing}, in the presence of
a break, the number of common factors increases across the changepoint: a
factor model with a changepoint can be written equivalently as a factor
model with no changepoint and an enlarged factor space (the second moment
matrix of the data, consequently, having more spiked eigenvalues).

On account of the considerations above, similarly to \citet{BT1} we propose
an on-line monitoring scheme based on sequentially monitoring the
eigenvalues of the second moment matrix of the data after a training period
in which no break occurs. In particular if, during the training period, $k$
common factors have been found, we monitor the $\left( k+1\right) $-th
eigenvalue: if no break occurs, this will be bounded over the monitoring
horizon, otherwise it will have a spike at some point in time. We are not
aware of any results on the second-order asymptotics of the estimated
eigenvalues, especially when these are non-spiked: indeed, in the context of
a factor model, \citet{trapani2018randomized} shows that the estimates of
the non-spiked eigenvalues may not even be consistent (see also %
\citealp{wang16}). Hence, we do not attempt to derive the asymptotic
distribution of the estimated eigenvalue used for the monitoring scheme.
Instead, we rely only on rates, and, at each point in time in the monitoring
horizon, we construct a transformation of the estimated $\left( k+1\right) $%
-th eigenvalue which - like the relevant eigenvalue - drifts to zero under
no changepoint, whereas it diverges to positive infinity under the
alternative of a changepoint. We then perturb these transformations by
adding a randomly generated \textit{i.i.d.} sequence with a standard normal
distribution. Under the null, such sequence remains an \textit{i.i.d.}
sequence with a standard normal distribution, whereas, under the alternative
of a changepoint, it has a drifting mean. Based on this, it is possible to
construct several monitoring schemes, based on the fluctuations of (weighted
or unweighted) partial sums (as in \citealp{Chu1996Monitoring}; %
\citealp{lajos04}; and \citealp{lajos07}), or on using the extreme value
distribution. In our analysis, it is important to have an estimate of the $%
\left( k+1\right) $-th eigenvalue whose rate is as small as possible under
no changepoint. Thus, when constructing the second moment matrix of the
data, we consider the use of \textit{projected} matrices %
\citep{Yu2021Projected}; we rely on a novel finding by \citet{hkyt}, where
almost sure rates of the estimated eigenvalues are derived and shown to be
wider than when using other approaches.

In our paper, we make at least three contributions. First, to the best of
our knowledge, this is the first contribution which studies online
changepoint detection in the context of matrix factor models, using
techniques specifically designed for this case. Second, we develop a
completely new randomisation scheme which is more natural than the one
developed in \citet{BT1}, where two randomisations are required involving
the choice of numerous tuning parameters. Third, we develop a substantial
methodological innovation to the online changepoint detection problem in
general. Within this context, existing procedures often rely on the CUSUM
process (or variations thereof, as e.g. the ones proposed by %
\citealp{kirch2017}; or \citealp{dettegosmann}), and detect a changepoint
when the fluctuations of the CUSUM exceed a boundary function. The choice of
such boundary function is however difficult and not unique, being, as %
\citet{Chu1996Monitoring} put it, \textquotedblleft often dictated by
mathematical convenience rather than optimality\textquotedblright\ (p.
1052). Conversely, in our case we base our boundary functions on the
\textquotedblleft natural\textquotedblright\ fluctuations of partial sums.
Thus, the theory developed in this paper can be extended to wide (and
diverse) variety of problems, well beyond the context of matrix factor
models.

\bigskip

The remainder of the paper is organized as follows. Section \ref{ap}
presents the setup assumptions and preliminaries. Section \ref{mt} gives the
methodologies and theories. Section \ref{fa} considered further
alternatives. Simulations are conducted in Section \ref{simulation} and
empirical studies are carried in Section \ref{empirical}. Section \ref%
{conclusion} concludes the paper. Assumptions and derivations, as well as further simulations, are relegated to the Supplement.

To end this section, we introduce some notations adopted in the paper. Given
an $n\times m$ matrix $A$, we denote its transpose as $A^{\prime }$ and its
element in position $\left( i,j\right) $ as $\left( A\right) _{ij}$, with $%
1\leq i\leq n$ and $1\leq j\leq m$; and we denote its $j$-th column as $%
A_{\cdot j}$ and its $i$-th row as $A_{i\cdot }$. When $n=m$, we denote the
eigenvalues of $A$ as $\lambda _{i}\left( A\right) $, sorted in decreasing
order. Throughout the paper, we often use the following sequence
\begin{equation}
l_{p_{1},p_{2},m}=\left( \frac{1}{p_{2}}+\frac{1}{m}+\frac{p_{1}}{\sqrt{%
mp_{2}}}\right) \left( \ln ^{2}p_{1}\ln p_{2}\ln m\right) ^{1+\epsilon }.
\label{lsequence}
\end{equation}%
Positive finite constants are denoted as $c_{0}$, $c_{1}$, ..., and their
values may change from line to line. Throughout the paper, we use the
short-hand notation \textquotedblleft a.s.\textquotedblright\ for
\textquotedblleft almost sure(ly)\textquotedblright . Given two sequences $%
a_{p_{1},p_{2},T}$ and $b_{p_{1},p_{2},T}$, we say that $%
a_{p_{1},p_{2},T}=o_{a.s.}\left( b_{p_{1},p_{2},T}\right) $ if, as $\min
\left\{ p_{1},p_{2},T\right\} \rightarrow \infty $, it holds that $%
a_{p_{1},p_{2},T}b_{p_{1},p_{2},T}^{-1}\rightarrow 0$ a.s.; we say that $%
a_{p_{1},p_{2},T}=O_{a.s.}\left( b_{p_{1},p_{2},T}\right) $ to denote that
as $\min \left\{ p_{1},p_{2},T\right\} \rightarrow \infty $, it holds that $%
a_{p_{1},p_{2},T}b_{p_{1},p_{2},T}^{-1}\rightarrow c_{0}<\infty $ a.s.; and
we use the notation $a_{p_{1},p_{2},T}=\Omega _{a.s.}\left(
b_{p_{1},p_{2},T}\right) $ to indicate that as $\min \left\{
p_{1},p_{2},T\right\} \rightarrow \infty $, it holds that $%
a_{p_{1},p_{2},T}b_{p_{1},p_{2},T}^{-1}\rightarrow c_{0}>0$ a.s.

\section{Model setup and spectra preliminaries\label{ap}}

In this section, we collect some preliminary results on the spectra of the
row (column) space of the second moment matrix of the data. We consider the
two-way factor structure model or matrix factor model:
\[
\left( X_{t}\right) _{p_{1}\times p_{2}}=\left( R\right) _{p_{1}\times
k_{1}}\left( F_{t}\right) _{k_{1}\times k_{2}}\left( C^{\prime }\right)
_{k_{2}\times p_{2}}+\left( E_{t}\right) _{p_{1}\times p_{2}}, \ k_1, \
k_2\geq0.
\]
In the presentation, we focus on detecting a change point in the \textit{row}
factor structure of $X_{t}$ - thus, the focus is on $R$, although all the
theory developed hereafter can be readily adapted to check for changes in $%
C^{\prime }$.

In order to fully make use of the two-way interactive factor structure, we
propose studying the spectrum of a projected column (row) covariance matrix,
as suggested by \citet{Yu2021Projected}. Heuristically, if $C$ is known and
satisfies the orthogonality condition $C^{\prime }C/p_{2}=I_{k_{2}}$, the
data matrix can be projected into a lower dimensional space by setting $%
Y_{t}=X_{t}C/p_{2}$. In view of this, we define
\begin{equation}
\widehat{M}_{1}=\frac{1}{m}\sum_{t=1}^{m}\widetilde{Y}_{t}\widetilde{Y}%
_{t}^{\prime },  \label{m1-hat}
\end{equation}%
where $\widetilde{Y}_{t}={{p_{2}^{-1}}}X_{t}\widetilde{C}$ and $\widetilde{C}
$ is an initial estimator of $C$. As suggested by \citet{Yu2021Projected},
the initial estimator can be set as $\widetilde{C}=\sqrt{p_{2}}Q$, where the
columns of $Q$ are the leading $k_{2}$ eigenvectors of $M_{r}$, where $M_{r}$
is the column \textquotedblleft flattened\textquotedblright\ sample
covariance matrices, i.e.,

\begin{eqnarray*}
M_{r}:= &&\frac{1}{Tp_{1}}\sum_{t=1}^{T}X_{t}^{\prime }X_{t}=\frac{1}{Tp_{1}}%
\sum_{t=1}^{T}\sum_{j=1}^{p_{1}}X_{j\cdot ,t}X_{j\cdot ,t}^{\prime }.
\end{eqnarray*}

Furthermore, if $k_{2} $ is not known, we can select the leading $\widetilde{%
k}$ eigenvectors of $M_{r}$ with $\widetilde{k}$\ chosen such that $%
\widetilde{k}\geq k_{2}$ (we refer to Section \ref{simulation} for details
on how we choose $\widetilde{k}$).

Let $\widehat{\lambda }_{j}$ denote the $j$-th largest eigenvalue of $%
\widehat{M}_{1}$, and recall that there are $k_{1}\geq 0$ row factors. The
following result, shown in \citet{hkyt}, measures the eigen-gap of $\widehat{%
M}_{1}$ (we refer to the supplement for Assumptions \ref{factors}-\ref{depFE}%
).

\begin{lemma}
\label{theorem:tildeM1} We assume that Assumptions \ref{factors}-\ref{depFE}
are satisfied. Then it holds that
\begin{equation}
\widehat{\lambda }_{j}=\Omega _{a.s.}\left( p_{1}\right) ,
\label{tildelarge}
\end{equation}%
for all $j\leq k_{1}$. Also, there exists a finite constant $c_{1}$ such
that
\begin{equation}
\widehat{\lambda }_{j}=o_{a.s.}\left( l_{p_{1},p_{2},m}\right) ,
\label{tildesmall}
\end{equation}%
for all $j>k_{1}$ and all $\epsilon >0$.
\end{lemma}

The eigen-gap in the spectrum of $\widehat M_1$ is the building block to
construct a testing procedure to decide whether a changpoint exist in the
monitoring horizon.

\section{Online changepoint detection\label{mt}}

We assume we have observed data over a training sample $1\leq t\leq m$, and
that no breaks were found.

\begin{assumpC}
\label{monitoring}\textbf{\ }The column space of $R$ does not change during $%
1\leq t\leq m$.
\end{assumpC}

Assumption \ref{monitoring} can be verified, and $k_{1}$ can be computed,
with the methodologies described in \citet{Yu2021Projected} and \citet{hkyt}%
. Note that we entertain the possibility that there is no factor structure
at all in the data during the training sample, i.e., $k_{1}=0$.

\bigskip

We monitor for changepoints in the row\ factor structure of $X_{t}$ across
an interval $m+1\leq t\leq m+T_{m}$, with - tidying up the notation - $%
m+T_{m}=T$. The monitoring schemes are based on the eigenvalue%
\[
\widehat{\lambda }_{k_{1}+1,\tau }=\lambda _{k_{1}+1}\left( \frac{1}{m}%
\sum_{t=\tau +1}^{m+\tau }\widetilde{Y}_{t}\widetilde{Y}_{t}^{\top }\right) ,
\]%
where $\widetilde{Y}_{t}=p_{2}^{-1}X_{t}\widetilde{C}${.}

In this section, we propose a method for the sequential detection of two
types of breaks; we extend our set-up, discussing further alternatives, in
Section \ref{fa}. A first source of change could be a scenario in which the
number of common factors is constant across regimes, but the row factor
space spanned by the columns of $R$ switches from one to another after a
point in time $t^{\ast }$, i.e.,
\begin{equation}
X_{t}=\left\{
\begin{array}{ll}
RF_{1,t}C^{\prime }+E_{t} & \text{for }1\leq t\leq m+t^{\ast } \\
\widetilde{R}F_{2,t}C^{\prime }+E_{t} & \text{for }t>m+t^{\ast }%
\end{array}%
\right. ,  \label{b1}
\end{equation}%
where $R=\left[ R_{0}|R_{1}\right] $, $\widetilde{R}=\left[ R_{0}|R_{2}%
\right] $, $R_{0}$ is a $p_{1}\times \left( k_{1}-c_{1}\right) $ matrix of
loadings which do not undergo a change, and $R_{1}$ and $R_{2}$ are $%
p_{1}\times c_{1}$ matrices of loadings which differ before and after the
changepoint $t^{\ast }$. We would like to point out that, in (\ref{b1}) and
henceforth, we assume that $C$ does not change merely for simplicity. Our
arguments, however, keep holding even in the presence of time-varying $C$
(see the discussion in Section \ref{fa}). In such a case, by Lemma \ref%
{theorem:tildeM1 copy(1)}\textit{(i)} in Appendix, it holds that
\begin{equation}
\widehat{\lambda }_{k_{1}+1,\tau }\left\{
\begin{array}{ll}
\leq c_{0}l_{p_{1},p_{2},m} & \text{for }\tau \leq t^{\ast } \\
\geq c_{1}\min \left( \frac{\tau -t^{\ast }}{m},\frac{m+t^{\ast }-\tau }{m}%
\right) p_{1} & \text{for }t^{\ast }<\tau <m+t^{\ast } \\
\leq c_{0}l_{p_{1},p_{2},m} & \text{for }\tau \geq m+t^{\ast }%
\end{array}%
\right. .  \label{b11}
\end{equation}%
As a second possible alternative, we consider the scenario whereby a set of
common factors appear after the point in time $t^{\ast }$, i.e., the column
space of $R$ enlarges:
\begin{equation}
X_{t}=\left\{
\begin{array}{ll}
RF_{a,t}C^{\prime }+E_{t} & \text{for }1\leq t\leq m+t^{\ast } \\
\widetilde{R}F_{t}C^{\prime }+E_{t} & \text{for }t>m+t^{\ast }%
\end{array}%
\right. ,  \label{b2}
\end{equation}%
where $\widetilde{R}=\left[ R|R_{3}\right] $, $R_{3}$ is a $p_{1}\times
c_{3} $ matrix, $F_{t}^{\prime }=\left[ F_{a,t}^{\prime }|F_{b,t}^{\prime }%
\right] $, and $F_{b,t}$ is a $c_{3}\times k_{2}$ matrix of new common
factors. Lemma \ref{theorem:tildeM1 copy(1)}\textit{(ii)} stipulates that,
in such a case%
\begin{equation}
\widehat{\lambda }_{k_{1}+1,\tau }\left\{
\begin{array}{ll}
\leq c_{0}l_{p_{1},p_{2},m} & \text{for }\tau \leq t^{\ast } \\
\geq c_{1}\frac{\tau -t^{\ast }}{m}p_{1} & \text{for }t^{\ast }<\tau
<m+t^{\ast } \\
\geq c_{0}p_{1} & \text{for }\tau \geq m+t^{\ast }%
\end{array}%
\right. .  \label{b21}
\end{equation}%
Based on (\ref{b11}) and (\ref{b21}), we propose a procedure for the null
hypothesis of no changepoint over the monitoring horizon in (\ref{fm1}).
Under the null, Lemma \ref{theorem:tildeM1} implies that
\begin{equation}
\widehat{\lambda }_{k_{1}+1,\tau }\leq c_{0}l_{p_{1},p_{2},m},
\label{mon-eig-null}
\end{equation}%
for all $1\leq \tau \leq T_{m}$. Under the alternative, at some point in
time $\tau >t^{\ast }$, it will hold that, for some positive $c_{0}$
\begin{equation}
\widehat{\lambda }_{k_{1}+1,\tau }\geq c_{0}p_{1}.  \label{mon-eig-alt}
\end{equation}%
Define now $0\leq \delta <1$, such that%
\begin{equation}
p_{1}^{-\delta }\widehat{\lambda }_{k_{1}+1,\tau }\left\{
\begin{array}{l}
o_{a.s.}\left( 1\right) \\
\geq c_{0}p_{1}^{1-\delta }\rightarrow \infty%
\end{array}%
\right.
\begin{array}{l}
\text{under no break} \\
\text{if there is a break}%
\end{array}%
.  \label{dichot}
\end{equation}%
The effect of this transformation is that $p_{1}^{-\delta }\widehat{\lambda }%
_{k_{1}+1,\tau }$ drifts to zero or infinity according as the null or the
alternative is true: any $\delta $ which satisfies this is, in principle,
admissible, and in Section \ref{fa} we discuss the selection of $\delta $ in
greater detail. According to (\ref{mon-eig-null}) and (\ref{mon-eig-alt}), $%
p_{1}^{-\delta }\widehat{\lambda }_{k_{1}+1,\tau }=o_{a.s.}\left( 1\right) $
as long as
\begin{equation}
\left\{
\begin{array}{ll}
\delta =\varepsilon & \ \text{if}\ \beta \leq 1/2 \\
\delta =1-1/(2\beta )+\varepsilon & \ \text{if}\ \beta >1/2%
\end{array}%
\right. ,  \label{equ:deltabeta}
\end{equation}%
where $\beta ={\ln p_{1}}/{\ln }\left( {p_{2}m}\right) $, and $\varepsilon
>0 $ is an arbitrarily small, user-defined number.

The dichotomous behaviour in (\ref{dichot}) is the building block of our
monitoring schemes, and it can be further enhanced by considering a
continuous transformation $g\left( \cdot \right) $ such that $%
\lim_{x\rightarrow 0}g\left( x\right) =0$ and $\lim_{x\rightarrow \infty
}g\left( x\right) =\infty $. Defining%
\begin{equation}
\psi _{\tau }=g\left( \frac{p_{1}^{-\delta }\widehat{\lambda }_{k_{1}+1,\tau
}}{p_{1}^{-1}\sum_{j=1}^{p_{1}}\widehat{\lambda }_{j,\tau }}\right) ,
\label{psi}
\end{equation}%
where $\widehat{\lambda }_{k_{1}+1,\tau }$ is normalised by the trace of $%
\left( \frac{1}{m}\sum_{t=\tau +1}^{m+\tau }\widetilde{Y}_{t}\widetilde{Y}%
_{t}^{\top }\right)$; again, other rescaling schemes are possible. By (\ref%
{dichot}), it follows that $\psi _{\tau }=o_{a.s.}\left( 1\right) $ for all $%
1\leq \tau \leq T_{m}$ under the null of no changepoint; conversely, in the
presence of a changepoint in $t^{\ast }$ it holds that $\psi _{\tau }\overset%
{a.s.}{\rightarrow }\infty $ for some $\tau >t^{\ast }$. Finally, after
computing $\psi _{\tau }$ at each $1\leq \tau \leq T_{m}$, we define the new
sequence $\left\{ y_{\tau },1\leq \tau \leq T_{m}\right\} $ as%
\begin{equation}
y_{\tau }=z_{\tau }+\psi _{\tau },  \label{yt}
\end{equation}%
where $z_{\tau }\overset{\mathit{i.i.d.}}{\sim} N\left( 0,1\right) $ for $%
1\leq \tau \leq T_{m}$. We allow the monitoring $T_{m}$ to go on for a long
time.

\begin{assumpC}
\label{horizon}It holds that $T_{m}=T_{m}\left( m\right) $ with $%
\lim_{m\rightarrow \infty }T_{m}=\infty $, and $T_{m}=\Omega \left(
m^{\varsigma }\right) $ with $\varsigma \geq 1$.
\end{assumpC}


\subsection{Monitoring schemes\label{online}}

We propose two (families of) monitoring schemes: one, more traditional,
based on the fluctuations of partial sums (Section \ref{ms1}), and one based
on the worst case scenario across the monitoring horizon (Section \ref{ms2}%
). In essence, there are two ways we could do monitoring: either considering
the partial sum process of $y_{\tau }$, or directly using the sequence $%
\left\{ y_{\tau },1\leq \tau \leq T_{m}\right\} $, noting that, under the
alternative, this will have a drift due to the presence of $\psi _{\tau }$.

In both cases, we will need the following restriction on the functional form
of $g\left( \cdot \right) $:

\begin{equation}
\lim_{\min \left( p_{1},T_{m}\right) \rightarrow \infty }T_{m}g\left(
p_{1}^{-\delta }l_{p_{1},p_{2},m}\right) =0.  \label{restriction}
\end{equation}%
This restriction, in essence, controls for a non-centrality parameter which
arises in the asymptotic of the test statistic under $H_{0}$. We note that (%
\ref{restriction}) does not contain any nuisance parameters: the sample
sizes $m$ and $p_{1}$, the monitoring horizon $T_{m}$, and $\delta $ and $%
g\left( \cdot \right) $ are all either given, or can be chosen so that (\ref%
{restriction}) holds.

Henceforth, we let $P^{\ast }$ denote the probability conditional on $%
\left\{ X_{t},1\leq t\leq T\right\} $; we use \textquotedblleft $\overset{%
P^{\ast }}{\rightarrow }$\textquotedblright , and \textquotedblleft $\overset%
{D^{\ast }}{\rightarrow }$\textquotedblright\ to denote convergence in
probability and in distribution according to $P^{\ast }$, respectively.

\subsubsection{Monitoring schemes based on partial sum processes\label{ms1}}

We begin with a family of \textquotedblleft traditional\textquotedblright\
monitoring schemes, based on partial sum processes. In particular, recall
the definition of $y_{\tau }$ in (\ref{yt}), and consider the partial sums
process%
\begin{equation}
S_{\tau }=\sum_{j=1}^{\tau }y_{j}.  \label{partial-sums}
\end{equation}%
Using the fluctuations of partial sums to detect changepoint is a standard
approach in this literature (see e.g. \citealp{Chu1996Monitoring}; and %
\citealp{lajos04}; \citealp{lajos07}). However, in our case we do not need
to derive the second order asymptotics of the eigenvalues of $\widehat{M}%
_{1} $, and we can rely on the \textquotedblleft natural\textquotedblright\
growth rate of the partial sum process defined in (\ref{partial-sums}).

We will use:

\begin{description}
\item[\textit{(i)}] the weighted functionals%
\begin{equation}
T_{m}^{\eta -1/2}\max_{1\leq \tau \leq T_{m}}\frac{\left\vert S_{\tau
}\right\vert }{\tau ^{\eta }},  \label{functionals}
\end{equation}%
for $0\leq \eta <1/2$;

\item[\textit{(ii)}] the standardised partial sums
\begin{equation}
\max_{1\leq \tau \leq T_{m}}\frac{\left\vert S_{\tau }\right\vert }{\tau
^{1/2}};  \label{functional-de}
\end{equation}

\item[\textit{(iii)}] the R\'{e}nyi statistics (see \citealp{horvathmiller})%
\begin{equation}
r_{T_{m}}^{\eta -1/2}\max_{r_{T_{m}}\leq \tau \leq T_{m}}\frac{\left\vert
S_{\tau }\right\vert }{\tau ^{\eta }},  \label{functionals-renyi}
\end{equation}%
for $\eta >1/2$, where $r_{T_{m}}$ is a sequence such that, as $%
T_{m}\rightarrow \infty $,
\[
r_{T_{m}}\rightarrow \infty \text{ \ and \ }\frac{r_{T_{m}}}{T_{m}}%
\rightarrow 0.
\]%
Define%
\[
\alpha _{T_{m}}=\sqrt{2\ln \ln T_{m}},\ \ \beta _{T_{m}}=2\ln \ln T_{m}+%
\frac{1}{2}\ln \ln \ln T_{m}-\frac{1}{2}\ln \pi .
\]
\end{description}

\begin{theorem}
\label{lemma:function} We assume that Assumptions \ref{factors}-\ref{depFE}
and \ref{monitoring}-\ref{horizon} hold, and that $g\left( \cdot \right) $
satisfies (\ref{restriction}). Then, under the null, as $\min \left(
m,p_{1},p_{2}\right) \rightarrow \infty $, it holds that

\begin{description}
\item[(i)]
\begin{equation}
T_{m}^{\eta -1/2}\max_{1\leq \tau \leq T_{m}}\frac{\left\vert S_{\tau
}\right\vert }{\tau ^{\eta }}\overset{D^{\ast }}{\rightarrow }\sup_{0\leq
u\leq 1}\frac{\left\vert W\left( u\right) \right\vert }{u^{\eta }},
\label{weighted}
\end{equation}%
for all $0\leq \eta <1/2$;

\item[(ii)]
\begin{equation}
P^{\ast }\left( \alpha _{T_{m}}\max_{1\leq \tau \leq T_{m}}\frac{\left\vert
S_{\tau }\right\vert }{\tau ^{1/2}}\leq v+\beta _{T_{m}}\right) =\exp \left(
-\exp \left( -v\right) \right) ,  \label{darling-erdos}
\end{equation}%
for all $-\infty <v<\infty $; and

\item[(iii)]
\begin{equation}
r_{T_{m}}^{\eta -1/2}\max_{r_{T_{m}}\leq \tau \leq T_{m}}\frac{\left\vert
S_{\tau }\right\vert }{\tau ^{\eta }}\overset{D^{\ast }}{\rightarrow }%
\sup_{1\leq u<\infty }\frac{\left\vert W\left( u\right) \right\vert }{%
u^{\eta }},  \label{renyi}
\end{equation}%
for all $\eta >1/2$. All the results hold for almost all realisations of $%
\{X_{t},1\leq t\leq T\}$.
\end{description}
\end{theorem}

The theorem provides the asymptotics under the null for the functionals
defined in (\ref{functionals})-(\ref{functionals-renyi}). In the case of
equation (\ref{weighted}), the null of no break is rejected whenever
\[
\left\vert S_{\tau }\right\vert >c_{\alpha ,\eta }T_{m}^{-\eta +1/2}\tau
^{\eta },
\]
where
\[
P\left( \sup_{0\leq u\leq 1}\frac{\left\vert W\left( u\right) \right\vert }{%
u^{\eta }}>c_{\alpha ,\eta }\right) =\alpha ,
\]%
thereby defining the changepoint estimate $\widehat{\tau }=\left\{ \min \tau
:\left\vert S_{\tau }\right\vert \geq c_{\alpha ,\eta }T_{m}^{-\eta
+1/2}\tau ^{\eta }\right\} $. The threshold $c_{\alpha ,\eta }T_{m}^{-\eta
+1/2}\tau ^{\eta }$ can be compared with other choices in the literature,
where such specification is \textquotedblleft often dictated by mathematical
convenience rather than optimality\textquotedblright\ (%
\citealp{Chu1996Monitoring}, p. 1052). In contrast, our methodology affords
to choose a \textquotedblleft natural\textquotedblright , non ad-hoc,
threshold. As mentioned in the introduction, \citet{BT1} study a similar
problem; however, the algorithm considered in that paper requires more
sophisticated assumptions and does not lend itself to the construction of a
natural threshold function like the ones considered here.

When using (\ref{functional-de}), equation (\ref{darling-erdos}) entails
that rejection takes place, with a changepoint being identified, at $%
\widehat{\tau }=\left\{ \min \tau :\left\vert S_{\tau }\right\vert \geq
c_{\alpha ,m}\tau ^{1/2}\right\} $, where the asymptotic critical value $%
c_{\alpha ,m}$ is defined as
\[
c_{\alpha ,m}=\frac{\beta _{Tm}-\ln \left( -\ln \left( 1-\alpha \right)
\right) }{\alpha _{Tm}}.
\]%
It is well known that convergence to the extreme value distribution is very
slow, and that the asymptotic critical value, in particular, overstates the
correct one, thus leading to lower power, even for large values of $T_{m}$
(see e.g. \citealp{csorgo1997}). \citet{gombay} provide approximate critical
values to overcome such issue.

Finally, in the case of (\ref{functionals-renyi}), rejection takes place at $%
\widehat{\tau }=\left\{ \min \tau :\left\vert S_{\tau }\right\vert \geq
c_{\alpha ,\eta }r_{T_{m}}^{-\eta +1/2}\tau ^{\eta }\right\} $, where%
\[
P\left( \sup_{1\leq u<\infty }\frac{\left\vert W\left( u\right) \right\vert
}{u^{\eta }}>c_{\alpha ,\eta }\right) =\alpha .
\]%
We point out that the scale transformation of the Wiener process yields%
\[
\sup_{1\leq u<\infty }\frac{\left\vert W\left( u\right) \right\vert }{%
u^{\eta }}\overset{D}{=}\sup_{0\leq s\leq 1}\frac{\left\vert W\left(
s\right) \right\vert }{s^{1-\eta }},
\]%
so that the critical values in Table 1 in \citet{lajos04} can be used here.

\bigskip

We now show that our procedures have power versus the alternatives
considered in (\ref{b1}) and (\ref{b2}).

\begin{assumpC}
\label{break} It holds that $t^{\ast }=O\left( m\right) $.
\end{assumpC}

In Assumption \ref{break}, we assume that, if a break occurs, this does not
happen too close to the end of the monitoring period.

\begin{theorem}
\label{function-power} We assume that Assumptions \ref{factors}-\ref{depFE}
and \ref{monitoring}-\ref{horizon} hold, and that $g\left( \cdot \right) $
satisfies (\ref{restriction}). Then, under (\ref{b1}) and (\ref{b2}), as $%
\min \left( m,p_{1},p_{2}\right) \rightarrow \infty $

\begin{description}
\item[(i)] if%
\begin{equation}
\left( \frac{t^{\ast }}{T_{m}}\right) ^{1-\eta }\frac{T_{m}^{1/2}}{\sqrt{\ln
\ln T_{m}}}g\left( p_{1}^{1-\delta }\right) \rightarrow \infty ,
\label{power-1}
\end{equation}%
it holds that%
\begin{equation}
T_{m}^{\eta -1/2}\max_{1\leq \tau \leq T_{m}}\frac{\left\vert S_{\tau
}\right\vert }{\tau ^{\eta }}\overset{P^{\ast }}{\rightarrow }\infty ,
\label{w-power}
\end{equation}%
for all $0\leq \eta <1/2$;

\item[(ii)] if
\begin{equation}
\left( \frac{t^{\ast }}{T_{m}}\right) ^{1/2}\frac{T_{m}^{1/2}}{\left( \ln
\ln T_{m}\right) }g\left( p_{1}^{1-\delta }\right) \rightarrow \infty ,
\label{power-2}
\end{equation}%
under (\ref{b1}) it holds that
\begin{equation}
P^{\ast }\left( \alpha _{T_{m}}\max_{1\leq \tau \leq T_{m}}\frac{\left\vert
S_{\tau }\right\vert }{\tau ^{1/2}}>v+\beta _{T_{m}}\right) =1,
\label{de-power}
\end{equation}%
for all $-\infty <v<\infty $;

\item[(iii)] if%
\begin{equation}
\left( \frac{r_{T_{m}}}{T_{m}}\right) ^{\eta -1/2}\left( \frac{t^{\ast }}{%
T_{m}}\right) ^{1-\eta }\frac{T_{m}^{1/2}}{\sqrt{\ln \ln T_{m}}}g\left(
p_{1}^{1-\delta }\right) \rightarrow \infty ,  \label{power-3}
\end{equation}%
under (\ref{b1}) it holds that%
\begin{equation}
\left( r_{T_{m}}\right) ^{\eta -1/2}\max_{r_{T_{m}}\leq \tau \leq T_{m}}%
\frac{\left\vert S_{\tau }\right\vert }{\tau ^{\eta }}\overset{P^{\ast }}{%
\rightarrow }\infty ,  \label{renyi-power}
\end{equation}%
for all $\eta >1/2$. All the results hold for almost all realisations of $%
\{X_{t},1\leq t\leq T\}$.
\end{description}
\end{theorem}

Theorem \ref{function-power} states that all our procedures have nontrivial
power versus changes. However, the choice of the weighing scheme affects
such power. Comparing (\ref{power-1}) with (\ref{power-2}), it is immediate
to see that standardised partial sum processes yield power versus
alternatives closer to the beginning of the monitoring period (i.e., for
smaller values of $t^{\ast }$); this is further enhanced in the case of R%
\'{e}nyi statistics, where breaks occurring at any $t^{\ast }>r_{T_{m}}$ are
detected.

\subsubsection{Monitoring schemes based on worst-case scenario\label{ms2}}

Recall that, heuristically, the sequence $\left\{ y_{\tau },1\leq \tau \leq
T_{m}\right\} $ is \textit{i.i.d.} Gaussian under the null, whereas it
diverges to positive infinity under the alternative. These features
(independence and Gaussianity) allow to propose a completely different
monitoring scheme based on%
\begin{equation}
Z_{T_{m}}=\max_{1\leq \tau \leq T_{m}}y_{\tau }.  \label{z}
\end{equation}%
In order to study the asymptotics of $Z_{T_{m}}$, we define the norming
sequences%
\[
a_{T_{m}}=\frac{b_{T_{m}}}{1+b_{T_{m}}^{2}},\ b_{T_{m}}=\sqrt{2\ln T_{m}}-%
\frac{\ln \ln T_{m}+\ln \left( 4\pi \right) }{2\sqrt{2\ln T_{m}}},
\]%
which are proposed in \citet{gasull}.

\begin{theorem}
\label{monitor-evt} We assume that Assumptions \ref{factors}-\ref{depFE} and %
\ref{horizon} hold, and that $g\left( \cdot \right) $ satisfies (\ref%
{restriction}). Then, under the null, as $\min \left( p_{1},p_{2},m\right)
\rightarrow \infty $, it holds that%
\begin{equation}
\lim_{\min \left( p_{1},p_{2},m\right) \rightarrow \infty }P^{\ast }\left(
\frac{Z_{T_{m}}-b_{T_{m}}}{a_{T_{m}}}\leq v\right) =\exp \left( -\exp \left(
-v\right) \right) ,  \label{null-dist-evt}
\end{equation}%
for almost all realisations of $\left\{ X_{t},1\leq t\leq T\right\} $ and
all $-\infty <v<\infty $. Under the alternatives (\ref{b1}) and (\ref{b2}),
if it holds that%
\begin{equation}
\frac{g\left( p_{1}^{1-\delta }\right) }{\sqrt{\ln T_{m}}}\rightarrow \infty
,  \label{alt-max}
\end{equation}%
as $\min \left( p_{1},m\right) \rightarrow \infty $, then it follows that%
\begin{equation}
\lim_{\min \left( p_{1},p_{2},m\right) \rightarrow \infty }P^{\ast }\left(
\frac{Z_{T_{m}}-b_{T_{m}}}{a_{T_{m}}}\leq v\right) =0,  \label{alt-dist-evt}
\end{equation}%
for almost all realisations of $\left\{ X_{t},1\leq t\leq T\right\} $ and
all $-\infty <v<\infty $.
\end{theorem}

According to (\ref{null-dist-evt}), rejection takes place, with a break
being identified, at $\widehat{\tau }=\left\{ \min \tau :y_{\tau }>c_{\alpha
,2}\right\} $, where the asymptotic critical value $c_{\alpha ,2}$ is given
by%
\[
c_{\alpha ,2}=b_{Tm}-a_{Tm}\ln \left( -\ln \left( 1-\alpha \right) \right) .
\]

Although convergence to the extreme value distribution is notoriously slow
(see \citealp{hall1979rate}, who finds a log rate), we found in our
simulations that the asymptotic critical value works very well as far as our
procedure is concerned. As an alternative, \textquotedblleft $M$-out-of-$N$%
\textquotedblright bootstrap or resampling schemes applied to $\left\{
z_{\tau },1\leq \tau \leq T_{m}\right\}$ could be considered (see e.g. %
\citealp{fukuchi}; and \citealp{politis1994}).

\section{Extension and Discussion\label{fa}}

\subsection{Further changepoint alternatives\label{further}}

Most of the focus of this paper has been on alternatives where the number of
factors increases after $t^{\ast }$. One alternative hypothesis which is
left out from this framework is the case where the column space of $R$
reduces, viz.
\begin{equation}
X_{t}=\left\{
\begin{array}{ll}
\widetilde{R}F_{t}C^{\prime }+E_{t} & \text{for }1\leq t\leq m+t^{\ast } \\
RF_{a,t}C^{\prime }+E_{t} & \text{for }t>m+t^{\ast }%
\end{array}%
\right. ,  \label{b3}
\end{equation}%
where $\widetilde{R}=\left[ R|R_{3}\right] $, $R_{3}$ is a $p_{1}\times
c_{3} $ matrix, $F_{t}^{\prime }=\left[ F_{a,t}^{\prime }|F_{b,t}^{\prime }%
\right] $, and $F_{b,t}$ is a $c_{3}\times k_{2}$ matrix of new common
factors. Here, we discuss how to extend our procedures proposed above to the
case of (\ref{b3}). Many arguments are either repetitive or straightforward
alterations of the procedures discussed above, and we therefore only provide
an overview of how our methodologies can be adapted to this case. By
standard arguments, under (\ref{b3}) it holds that%
\begin{equation}
\widehat{\lambda} _{k_{1},\tau }\left\{
\begin{array}{ll}
\geq c_{0}p_{1} & \text{for }\tau \leq t^{\ast } \\
\leq c_{1}\left( \frac{\tau -t^{\ast }+m}{m}\right) p_{1} & \text{for }%
t^{\ast }<\tau <m+t^{\ast } \\
\leq c_{0}l_{p_{1},p_{2},m} & \text{for }\tau \geq m+t^{\ast }%
\end{array}%
\right. .  \label{b3111}
\end{equation}

Based on the results discussed above, it is easy to show that under the null
of no change, it holds that
\[
p_{1}^{-\delta }\widehat{\lambda }_{k_{1},\tau }=\Omega _{a.s.}\left(
p_{1}^{1-\delta }\right) ,
\]%
where $\delta $ is defined in (\ref{equ:deltabeta}). In the presence of a
break, by (\ref{b3111}) it holds that there exists a $\widetilde{\tau }%
>t^{\ast }$ such that, for $\tau \geq \widetilde{\tau }$, $p_{1}^{-\delta }%
\widehat{\lambda }_{k_{1},\tau }=o_{a.s.}\left( 1\right) $. Consider now a
continuous transformation $\widetilde{g}\left( \cdot \right) $ such that $%
\lim_{x\rightarrow 0}\widetilde{g}\left( x\right) =0$ and $%
\lim_{x\rightarrow \infty }\widetilde{g}\left( x\right) =\infty $, e.g. $%
\widetilde{g}\left( x\right) =\exp \left( 1/x\right) -1$. Define%
\begin{equation}
\widetilde{\psi }_{\tau }=\widetilde{g}\left( \frac{p_{1}^{-\delta }\widehat{%
\lambda }_{k_{1},\tau }}{p_{1}^{-1}\sum_{j=1}^{p_{1}}\widehat{\lambda }%
_{j,\tau }}\right) ,  \label{psi-tilde}
\end{equation}%
where $\widehat{\lambda }_{k_{1},\tau}$ is normalised by the trace of $%
\left( \frac{1}{m}\sum_{t=\tau +1}^{m+\tau }\widetilde{Y}_{t}\widetilde{Y}%
_{t}^{\top }\right)$, but again other rescalings are possible. By
continuity, it follows that
\[
\widetilde{\psi }_{\tau }=o_{a.s.}\left( 1\right) \text{ for all }1\leq \tau
\leq T_{m}\text{ under the null of no changepoint,}
\]%
and%
\[
\widetilde{\psi }_{\tau }\overset{a.s.}{\rightarrow }\infty \text{ for some }%
\tau \geq \widetilde{\tau }>t^{\ast }\text{, in the presence of a
changepoint in }t^{\ast }.
\]%
From hereon, the procedures described above can be applied, using $%
\widetilde{\psi }_{\tau }$ in lieu of $\psi _{\tau }$, with the same results.

\bigskip

Secondly, we note that we have discussed our procedure for the case where $C$
is constant over time. This simplifies the presentation, but it may be
considered an unrealistic set-up. We now provide some arguments which show
that our procedure is able to detect changes in $R$ also in the possible
presence of changes in $C$. Inspired by \citet{baltagi2015}, we consider the
following scenario%
\begin{equation}
X_{t}=\left\{
\begin{array}{ll}
R_{1}F_{t}C_{t}^{\prime }+E_{t} & \text{for }1\leq t\leq m+t^{\ast } \\
R_{2}F_{t}C_{t}^{\prime }+E_{t} & \text{for }t>m+t^{\ast }%
\end{array}%
\right. ,  \label{Cchanges}
\end{equation}%
where we allow for%
\[
C_{t}=\left\{
\begin{array}{ll}
C_{1} & \text{for }1\leq t\leq m+\widetilde{t} \\
C_{2} & \text{for }t>m+\widetilde{t}%
\end{array}%
\right. ,
\]%
and - only for simplicity - we assume that: $\widetilde{t}<t^{\ast }$, $%
C_{1} $ and $C_{2}$ have the same number of columns $k_{2}$ and full rank.
Define now the (full rank) $k_{3}\times k_{3}$ matrix $C_{3}$ (where $%
k_{3}\geq k_{2}$) such that%
\[
C_{1}=C_{3}\Omega _{1}\text{ \ and \ }C_{2}=C_{3}\Omega _{2},
\]%
where $\Omega _{1}$ and $\Omega _{2}$ are $k_{3}\times k_{2}$ matrices with
full rank $k_{2}$. The equation above stipulates that the column spaces of
both $C_{1}$ and $C_{2}$ are spanned by the column space of $C_{3}$. In
particular, if $C_{1}$ and $C_{2}$ are orthogonal, then $k_{3}=2k_{2}$; if
they both lie in the same colum spaces, $k_{3}=k_{2}$; and, in general, $%
k_{2}\leq k_{3}\leq 2k_{2}$. Hence we can rewrite (\ref{Cchanges}) as%
\[
X_{t}=\left\{
\begin{array}{ll}
R_{1}F_{t}\Omega _{1}^{\prime }C_{3}^{\prime }+E_{t} & \text{for }1\leq
t\leq m+\widetilde{t} \\
R_{1}F_{t}\Omega _{2}^{\prime }C_{3}^{\prime }+E_{t} & \text{for }m+%
\widetilde{t}<t\leq m+t^{\ast } \\
R_{2}F_{t}\Omega _{2}^{\prime }C_{3}^{\prime }+E_{t} & \text{for }%
t>m+t^{\ast }%
\end{array}%
\right. .
\]%
Finally, letting%
\[
\widetilde{F}_{t}=\left\{
\begin{array}{ll}
F_{t}\Omega _{1}^{\prime } & \text{for }1\leq t\leq m+\widetilde{t} \\
F_{t}\Omega _{2}^{\prime } & \text{for }m+\widetilde{t}<t\leq m+t^{\ast } \\
F_{t}\Omega _{2}^{\prime } & \text{for }t>m+t^{\ast }%
\end{array}%
\right. ,
\]%
we can write%
\begin{equation}
X_{t}=\left\{
\begin{array}{ll}
R_{1}\widetilde{F}_{t}C_{3}^{\prime }+E_{t} & \text{for }1\leq t\leq
m+t^{\ast } \\
R_{2}\widetilde{F}_{t}C_{3}^{\prime }+E_{t} & \text{for }t>m+t^{\ast }%
\end{array}%
\right. .  \label{cchanges-2}
\end{equation}%
Upon extending Assumption \ref{factors}\textit{(ii)} by requiring that, as $%
m\rightarrow \infty $%
\[
\frac{1}{m}\sum_{t=1}^{m}F_{t}AF_{t}^{\prime }\overset{P}{\rightarrow }%
\widetilde{\Sigma },
\]%
with $\widetilde{\Sigma }$ positive definite for all positive definite
matrices $A$, it is possible to use the theory developed in \cite%
{Yu2021Projected} to estimate $C_{3}$, thereby obtaining the same results as
above.

\subsection{Detection delay\label{delay}}

In this section, we report some heuristic considerations on how the
detection delay is affected by the combinations of $m$, $p_{1}$ and $p_{2}$.
In particular, we note that the size of the rolling window $m$ is typically
user-defined; hence, our analysis offers some guidelines as to the choice of
the rolling window, $m$, complementing the findings from our simulations.

Recall that, by (\ref{b11}) and (\ref{b21}), in the presence of a
changepoint at time $t^{\ast }+1$, the $\left( k+1\right) $-th eigenvalue
calculated at $\tau \geq t^{\ast }+1$ - denoted as $\widehat{\lambda }%
_{k_{1}+1,\tau }$\ - will be proportional to $\frac{\tau -t^{\ast }}{m}p_{1}$%
. In the construction of our test statistics, we need to premultiply $%
\widehat{\lambda }_{k_{1}+1,\tau }$ by $p_{1}^{-\delta }$, to ensure that,
when there is no break, $p_{1}^{-\delta }\widehat{\lambda }_{k_{1}+1,\tau }$
vanishes. This entails that, heuristically, a changepoint will be detected
as long as%
\begin{equation}
\frac{p_{1}^{1-\delta }}{m}\left( \tau -t^{\ast }\right) \rightarrow \infty .
\label{general}
\end{equation}%
As in (\ref{equ:deltabeta}), it may be convenient to consider separately the
cases where $p_{1}$ is \textquotedblleft large\textquotedblright\ compared
to $\left( mp_{2}\right) ^{1/2}$, and the case where it is \textquotedblleft
smaller\textquotedblright\ than $\left( mp_{2}\right) ^{1/2}$.

Whenever $p_{1}=o\left( \left( mp_{2}\right) ^{1/2}\right) $, we note from (%
\ref{equ:deltabeta}) that, essentially, we can use $\delta =0$. In this
case, detection will take place as long as%
\begin{equation}
\frac{p_{1}}{m}\left( \tau -t^{\ast }\right) \rightarrow \infty .
\label{small-1}
\end{equation}%
In this case, a natural choice for the size of the rolling window would be $%
m=O\left( p_{1}^{1-\varepsilon }\right) $, for any $\varepsilon >0$, or even
$m=O\left( {p_{1}}/{\ln p_{1}}\right) $. With this value of $m$, (\ref%
{small-1}) will hold as long as $\tau -t^{\ast }>0$, thus ensuring a short
delay in the detection time of a chagepoint. On the other hand, when $\left(
mp_{2}\right) ^{1/2}=o\left( p_{1}\right) $, (\ref{equ:deltabeta}) requires
a choice of $\delta >0$ to ensure that $p_{1}^{1-\delta }=o\left( \left(
mp_{2}\right) ^{1/2}\right) $. Considering the case where $p_{1}^{1-\delta }$
is exactly of order $\left( mp_{2}\right) ^{1/2}$, detection of a
changepoint will take place as long as%
\begin{equation}
\frac{p_{1}^{1-\delta }}{m}\left( \tau -t^{\ast }\right) =c_{0}\frac{\left(
mp_{2}\right) ^{1/2}}{m}\left( \tau -t^{\ast }\right) =c_{1}\left( \frac{%
p_{2}}{m}\right) ^{1/2}\left( \tau -t^{\ast }\right) \rightarrow \infty ,
\label{small-2}
\end{equation}%
which would suggest the choice $m=o\left( p_{2}\right) $ to ensure that
detection takes place after a finite number of time periods.

\section{Simulations Study\label{simulation}}


In this section, we assess - through synthetic data - the performance of the
proposed procedures on testing and locating change points. We begin by
considering the main set-up of this paper, namely studying size and power in
the presence of a changepoint which leads to an increase in the number of
common factors between the pre- and post-break regimes.

Throughout the section, the data generating process is similar to the one
used in \cite{Yu2021Projected}. Specifically, under the null hypothesis $%
H_{0}$ without change point, we set the row/column factor numbers $%
k_{1}=k_{2}=3$. The entries of $R$ and $C$ are independently sampled from
uniform distribution $\mathcal{U}(-\sqrt{3},\sqrt{3})$, while
\begin{equation}
\begin{split}
\text{Vec}(F_{t})=& \phi \times \text{Vec}(F_{t-1})+\sqrt{1-\phi ^{2}}\times
\epsilon _{t},\quad \epsilon _{t}\overset{i.i.d.}{\sim }\mathcal{N}({\ 0}%
,I_{k_{1}\times k_{2}}), \\
\text{Vec}(E_{t})=& \psi \times \text{Vec}(E_{t-1})+\sqrt{1-\psi ^{2}}\times
\text{Vec}(U_{t}),\quad U_{t}\overset{i.i.d.}{\sim }\mathcal{MN}({\ 0}%
,U_{E},V_{E}),
\end{split}
\label{equ4.1}
\end{equation}%
where $U_{t}$ is from a matrix-normal distribution, i.e., $\text{Vec}(U_{t})%
\overset{i.i.d.}{\sim }\mathcal{N}({\ 0},V_{E}\otimes U_{E})$. $U_{E}$ and $%
V_{E}$ are matrices with ones on the diagonal, and the off-diagonal entries
are $1/p_{1}$ and $1/p_{2}$, respectively. The parameters $\phi $ and $\psi $
controls both the temporal and cross-sectional correlations of $X_{t}$. In
the simulation study, we let $\psi =\phi =0.1$ and set $T=200$, $p_{1}\in
\{50,80,100\}$ and $p_{2}\in \{20,50,80\}$. The monitoring procedures are
based on the $4$-th largest eigenvalue of the rolling column-column sample
covariance matrix. When calculating the initial projection matrix $%
\widetilde{C}$, we always use $k_{2}=k_{\max }=8$. To calculate the
sequences $\psi _{\tau }$, we let $\epsilon =0.05$ in (\ref{equ:deltabeta}),
$g(x)=[\exp (x)-1]^{4}$ in (\ref{psi}) while $m\in \{50,80,100\}$. All
results have been obtained using $1,000$ replications.

\begin{table*}[hbtp]
\caption{Empirical sizes ($\%$) under $H_0$ over 1000 replications. $T=200$,
$k_1=k_2=3$, $\protect\epsilon=0.05$.}
\label{tab1}
\begin{center}
{\small \ \addtolength{\tabcolsep}{0pt} \renewcommand{\arraystretch}{1.2}
\scalebox{0.75}{ 	
			\begin{tabular*}{21cm}{ccccccccccccccc}
				\toprule[1.2pt]
				&&&\multicolumn{6}{l}{$\alpha=0.05$}	&\multicolumn{6}{l}{$\alpha=0.10$}\\\cmidrule(lr){4-9}\cmidrule(lr){10-15}
				&&&\multicolumn{5}{l}{Partial-sum}&Worst&\multicolumn{5}{l}{Partial-sum}&Worst\\\cmidrule(lr){4-8}\cmidrule(lr){10-14}
				$m$&$p_1$&$p_2$&$\eta=0$&$\eta=0.25$&$\eta=0.5$&$\eta=0.65$&$\eta=0.75$	&case&$\eta=0$&$\eta=0.25$&$\eta=0.5$&$\eta=0.65$&$\eta=0.75$	&case	\\\midrule[1.2pt]
50&50&20&5.7&5.2&2.7&3.9&3.5&3.6&9.6&9.6&7.2&7.2&6.7&9.4
\\
50&50&50&4.1&4.4&1.6&2.7&2.8&5&9.2&8.9&6.4&4.9&4.7&10.1
\\
50&50&80&5.1&4.4&2&3.9&3.3&4.1&9.5&9.8&6.1&6.4&6.5&8.6
\\
50&80&20&4.5&3.9&1.7&3.1&3.1&4.4&8.8&9.2&4.7&6.4&6.2&8.3
\\
50&80&50&5.3&4.8&1.9&3.7&3.2&5.1&9.6&9.7&6.2&7.2&7.5&9.7
\\
50&80&80&4.5&4.5&1.7&2.9&2.9&4.5&9&9.3&6.5&5.5&6.3&9.8
\\
50&100&20&4.6&4.9&1.4&2.3&2.2&4.1&10.4&10.2&4.8&5.7&5.6&8.8\\
50&100&50&4&3.9&1.3&2.2&1.8&3.3&8.3&9&4.3&4.8&4.9&8.7
\\
50&100&80&5.4&5.8&1.5&2&2&4.2&9.1&9.4&5.9&6&6.5&9.7
\\\midrule[1.2pt]
80&50&20&4.8&4.1&2.2&3.1&3.2&3.3&8.6&8&6.2&6.5&6.3&8.3\\
80&50&50&3.9&3.8&1.7&2.7&3&3.7&9.3&9.5&5.7&6.8&6&8.1
\\
80&50&80&4&4.3&1.5&3.3&3.3&3&9.3&8.9&5.8&6.8&7.2&8.4
\\
80&80&20&4.5&4.2&2.1&2.9&2.7&3.4&8.8&8.5&5.1&6.8&6.2&7.9\\
80&80&50&4.4&4.1&1.6&4&3.5&5.2&9.2&9.2&5.3&6.9&7.2&10.8
\\
80&80&80&4.6&4.4&1.6&3.5&3.2&3.3&10.4&9.4&6.7&6.7&6.8&8.4
\\
80&100&20&3.5&3.4&0.9&2.4&2.7&4.4&9.2&7.9&4.2&5.4&6.1&9
\\
80&100&50&4.1&3.7&1.8&2.9&2.9&4.5&9&8&5.3&6.4&5.6&9.2
\\
80&100&80&4.2&3.9&1.5&2.4&2.1&4&7.9&7.7&4.8&6.3&6.6&9.9
\\\midrule[1.2pt]
100&50&20&5.2&4.5&1.8&2.4&2.5&4.3&10.3&9.9&6.2&5.9&5.6&8.6\\
100&50&50&4.1&4.2&2.3&3.2&2.9&4&7.9&7.7&5.2&6.2&5.4&9.3
\\
100&50&80&4.3&3.7&1.2&3.2&3.1&4.1&8.3&7.5&4.5&5.6&5.9&8.5
\\
100&80&20&4.4&4&1.4&2.3&2.3&4.3&10.3&9.2&4.9&6.2&5.7&9.7
\\
100&80&50&4.2&4.2&1.6&2.2&2&4.2&7.8&8.1&5.5&5.3&5.6&9.1
\\
100&80&80&4.4&3.9&1&2.8&3.5&4.4&8.8&7.7&4.7&6&6.4&9.8
\\
100&100&20&4.8&3.9&1.4&2.9&2.2&3.3&7.7&8.2&5.5&6&5.8&6.8\\
100&100&50&4.8&4.7&1.5&3.1&3.7&4.7&9.1&8.4&5.7&6.7&7&9.3
\\
100&100&80&4.4&4.5&2.3&3.2&3.3&3.6&9.6&9.5&5.9&6.8&7.2&9.2\\
				\bottomrule[1.2pt]		
		\end{tabular*}} }
\end{center}
\end{table*}

Table \ref{tab1} reports the empirical rejection frequencies under the null,
for various values of $\eta $ when using procedures based on partial sums.
It is worth pointing out that, in the context of online changepoint
detection, size control is different than in the standard Neyman-Pearson
testing context; as \citet{lajos07} put it, \textquotedblleft the goal is to
keep the probability of false rejection below $\alpha $ rather than to make
it close to $\alpha $\textquotedblright . We note that, in our experiments,
all the empirical sizes are controlled at the given significance level, even
in the case of small sample (very few exceptions are encountered when $m$ is
large). Specifically, the empirical sizes are closer to their theoretical
significance levels when using the partial-sum method with $\eta <0.5$, and
when using the worst-case method. The latter finding is interesting, since
convergence to the extreme value distribution is notoriously slow (%
\citealp{hall1979rate}). When $\eta =0.5$, the empirical sizes are usually
smaller than the significance levels $\alpha $, as also noted in %
\citet{gombay}. Similarly, when $\eta >0.5$, the empirical sizes also tends
to be smaller than the theoretical level $\alpha $, mainly because in this
case (similarly to the Darling-Erd\H{o}s case where $\eta =0.5$) the
effective sample size that determines the asymptotic distribution is smaller
than $T_{m}$; similar results have also been observed in \citet{ht2021} in
the context of in-sample changepoint detection.

We now study the power of our procedure to detect changepoints. As a first
alternative, we consider (\ref{b1}) - i.e. the case where the loadings
change after a change point located at $t^{\ast }=0.5T$
\[
X_{t}=\left\{ \begin{matrix} & RF_{t}C^{\prime }+E_{t}, & 1\leq t\leq
t^{\ast }, \\ & R_{new}F_{t}C^{\prime }+E_{t}, & t^{\ast }+1\leq t\leq
T,\end{matrix}\right.
\]%
where $R_{new}$ is regenerated after time point $t^{\ast }+1$, also with
\textit{i.i.d.} entries from $\mathcal{U}(-\sqrt{3},\sqrt{3})$. That is, the
loading space changes after the change point. We would like to point out
that, in this set of experiments and in all other experiments, the empirical
rejection frequencies under alternatives are all equal to 1 - in essence,
this entails that our procedures will always, eventually, find evidence of a
changepoint if present. Median detection delays are reported in Table \ref%
{tab3}: based on those results, we conclude that our procedures tend to have
short detection delays in all scenarios considered, even when the
cross-sectional dimensions $p_{1},p_{2}$ are small, offering accurate and
early detection. Results improve as $p_{2}$ increases, and they also seem to
marginally worsen as $m$ increases. Both findings corroborate our
conclusions in Section \ref{delay}. Interestingly, detection delays are
shorter, albeit marginally, when using monitoring schemes based on the
worst-case scenario, making the case for this methodology.

\begin{table*}[tbph]
\caption{Median delays for detecting change point under \protect\ref{b1}
(loading space changes) over 1000 replications. $T=200$, $k_{1}=k_{2}=3$, $%
\protect\epsilon =0.05$.}
\label{tab3}
\begin{center}
{\small \ \addtolength{\tabcolsep}{0pt} \renewcommand{\arraystretch}{1.2}
\scalebox{0.75}{ 	
			\begin{tabular*}{21cm}{ccccccccccccccc}
				\toprule[1.2pt]
				&&&\multicolumn{6}{l}{$\alpha=0.05$}	&\multicolumn{6}{l}{$\alpha=0.10$}\\\cmidrule(lr){4-9}\cmidrule(lr){10-15}
				&&&\multicolumn{5}{l}{Partial-sum}&Worst&\multicolumn{5}{l}{Partial-sum}&Worst\\\cmidrule(lr){4-8}\cmidrule(lr){10-14}
				$m$&$p_1$&$p_2$&$\eta=0$&$\eta=0.25$&$\eta=0.5$&$\eta=0.65$&$\eta=0.75$	&case&$\eta=0$&$\eta=0.25$&$\eta=0.5$&$\eta=0.65$&$\eta=0.75$	&case	\\\midrule[1.2pt]
50&50&20&5&5&5&5&5&4&5&5&5&5&5&4
\\
50&50&50&3&3&3&3&3&2&3&3&3&3&3&2
\\
50&50&80&3&3&3&3&3&2&3&3&3&3&3&2
\\
50&80&20&5&5&5&5&5&4&5&5&5&5&5&3
\\
50&80&50&3&3&3&3&3&2&3&3&3&3&3&2
\\
50&80&80&2.5&2&2&2&3&2&2&2&2&2&2&2
\\
50&100&20&5&5&5&5&5&4&5&5&5&5&5&4
\\
50&100&50&3&3&3&3&3&2&3&3&3&3&3&2
\\
50&100&80&3&2&2&2&3&2&2&2&2&2&2&2
\\\midrule[1.2pt]
80&50&20&6&6&6&6&6&5&6&5&5&5&5&4
\\
80&50&50&5&5&4&4&5&4&5&4&4&4&4&4
\\
80&50&80&5&5&4&4&5&4&5&4&4&4&4&4
\\
80&80&20&6&6&6&6&6&5&6&6&5&5&6&4
\\
80&80&50&4&4&4&4&4&3&4&4&3&3&3&3
\\
80&80&80&3&3&3&3&3&2&3&3&3&3&3&2
\\
80&100&20&6&6&6&6&6&5&6&6&6&6&6&5
\\
80&100&50&4&4&4&4&4&3&4&4&3&3&4&3
\\
80&100&80&3&3&3&3&3&2&3&3&3&3&3&2
\\\midrule[1.2pt]
100&50&20&7&6&5&5&5&5&6&6&5&5&5&5
\\
100&50&50&6&5&5&5&5&5&6&5&5&5&4&5
\\
100&50&80&6&5&5&4&4&4&6&5&5&4&4&4
\\
100&80&20&7&6&6&5&5&5&7&6&5&5&5&5
\\
100&80&50&4&4&3&4&4&3&4&4&3&4&4&3
\\
100&80&80&4&3&3&4&4&3&4&3&3&4&4&3
\\
100&100&20&7&6&6&5&5&5&7&6&6&5&5&5
\\
100&100&50&4&4&3&4&4&3&4&4&3&4&4&3
\\
100&100&80&3&3&3&4&4&3&3&3&3&4&4&2.5\\
				\bottomrule[1.2pt]		
		\end{tabular*}} }
\end{center}
\end{table*}

We now turn to considering the alternative (\ref{b2}) - i.e. the case where
the number of common factors increases after the change point. Specifically,
we generate data according to
\[
X_{t}=\left\{ \begin{aligned} &R F_tC^\prime+E_t,& 1\le t\le \tau,\\
&RF_tC^\prime+\ell \tilde f_t C^\prime +E_t,&\tau+1\le t\le T, \end{aligned}%
\right. ,
\]%
where $\ell $ is a $p_{1}\times 1$ vector with entries from \textit{i.i.d.}
uniform distribution $\mathcal{U}(-\sqrt{3},\sqrt{3})$ and $\tilde{f}_{t}$
are the additional $1\times k_{2}$ factor scores from \textit{i.i.d.}
standard normal distributions. Therefore, after the change point, the number
of row factors grows to $4$, and we are monitoring according to the $4$-th
eigenvalue of the rolling column-column sample covariance matrix. The other
parameters are set exactly the same as those introduced above.

\begin{table*}[hbtp]
\caption{Median delays for detecting change point under \protect\ref{b2}
(number of row factors increases) over 1000 replications. $T=200$, $%
k_1=k_2=3 $, $\protect\epsilon=0.05$.}
\label{tab4}
\begin{center}
{\small \ \addtolength{\tabcolsep}{0pt} \renewcommand{\arraystretch}{1.2}
\scalebox{0.75}{ 	
			\begin{tabular*}{21cm}{ccccccccccccccc}
				\toprule[1.2pt]
				&&&\multicolumn{6}{l}{$\alpha=0.05$}	&\multicolumn{6}{l}{$\alpha=0.10$}\\\cmidrule(lr){4-9}\cmidrule(lr){10-15}
				&&&\multicolumn{5}{l}{Partial-sum}&Worst&\multicolumn{5}{l}{Partial-sum}&Worst\\\cmidrule(lr){4-8}\cmidrule(lr){10-14}
				$m$&$p_1$&$p_2$&$\eta=0$&$\eta=0.25$&$\eta=0.5$&$\eta=0.65$&$\eta=0.75$	&case&$\eta=0$&$\eta=0.25$&$\eta=0.5$&$\eta=0.65$&$\eta=0.75$	&case	\\\midrule[1.2pt]
50&50&20&8&8&8&8&8&6&8&8&8&8&8&6
\\
50&50&50&5&5&5&5&5&4&5&5&5&5&5&4
\\
50&50&80&5&5&5&5&5&4&5&5&5&5&5&4
\\
50&80&20&8&8&8&8&8&6&8&8&8&8&8&6
\\
50&80&50&5&5&5&5&6&4&5&5&5&5&5&4
\\
50&80&80&4&4&4&4&4&3&4&4&4&4&4&3
\\
50&100&20&8&8&8&8&8&6&8&8&8&8&8&6
\\
50&100&50&5&5&5&5&6&4&5&5&5&5&5&4
\\
50&100&80&4&4&4&4&4&3&4&4&4&4&4&3
\\\midrule[1.2pt]
80&50&20&10&9&9&9&9&8&9&9&9&9&9&7
\\
80&50&50&8&8&7&7&8&6&8&7&7&7&7&6
\\
80&50&80&8&7&7&7&7&6&8&7&7&7&7&6
\\
80&80&20&9.5&9&9&9&9&8&9&9&9&9&9&8
\\
80&80&50&6&6&6&6&6&5&6&6&6&6&6&5
\\
80&80&80&5&5&5&5&5&4&5&5&5&5&5&4
\\
80&100&20&10&9&9&9&9&8&9&9&9&9&9&8\\
80&100&50&6&6&6&6&6&5&6&6&6&6&6&5
\\
80&100&80&5&5&5&5&5&4&5&5&5&5&5&4
\\\midrule[1.2pt]
100&50&20&10&9&9&9&9&8&10&9&9&8&8&8\\
100&50&50&9&9&8&8&8&8&9&8&8&7&7&8
\\
100&50&80&9&8&8&8&8&7&9&8&8&7&7&7
\\
100&80&20&10&10&9&9&9&9&10&9&9&9&9&8\\
100&80&50&7&6&6&6&6&5&7&6&6&5&5&5
\\
100&80&80&6&6&5&5&5&5&6&6&5&5&5&5
\\
100&100&20&10&10&9&9&9&9&10&9&9&9&9&8\\
100&100&50&7&6&6&6&6&6&7&6&6&6&5&6
\\
100&100&80&6&5&5&4&4&4&6&5&5&4&4&4\\
				\bottomrule[1.2pt]		
		\end{tabular*}} }
\end{center}
\end{table*}

Median detection delays are reported in Table \ref{tab4}. The delays appear
to be slightly larger than those reported for alternative hypothesis (\ref%
{b1}) in Table \ref{tab3}. However, the monitoring procedures can still
provide accurate and timely detection when change occurs under (\ref{b2}),
even in small samples cases. Note that the impact of the sample sizes $m$, $%
p_{1}$ and $p_{2}$ is exactly the same as under (\ref{b1}).

\bigskip

In the Supplement, we report further simulations under different scenarios.
Specifically, in Tables \ref{tab5} and \ref{tab6}, we assess the robustness
of our methdologies in the presence of changes in $C$, noting that results
are virtually unchanged in that case. We assess the sensitivity of our
procedures to $T$\ and $k_{2}$, and to $\delta $, in Tables \ref{tab8} and %
\ref{tab9} respectively. In both cases, size control is (marginally)
affected, especially when varying $T$\ and $k_{2}$, whereas detection delays
are essentially the same.

\section{Empirical application\label{applic}}

We validate our methodology through two empirical applications: in Section %
\ref{empirical}, we consider a $10\times 10$ matrix of portfolio returns; in
Section \ref{mmi}, we consider an application to macro data.

\subsection{ Fama-French 100 portfolios\label{empirical}}

In this section, we illustrate the usefulness of our sequential monitoring
schemes through an application to financial data. We use the Fama and French
$10\times 10$ series, which has been considered in several applications in
the context of matrix factor models (see e.g. \citealp{wang2019factor}; %
\citealp{Yu2021Projected}; and \citealp{chenthreshold}). The dataset
comprises monthly market-adjusted return series, with portfolios being the
intersections of $10$ portfolios formed by size (market equity, ME) and $10$
portfolios formed by the ratio of book equity to market equity (BE/ME),
which leads to $10\times 10$ matrix-variate observations.\footnote{%
Data have been downloaded from %
\url{http://mba.tuck.dartmouth.edu/pages/faculty/ken.french/data_library.html}%
.} We collect portfolio series from January 1964 to December 2020, totalling
$684$ months. Missing values (missing rate is $0.25\%$) are inputed by
linear interpolation; subsequently, following the preprocessing procedures
in \citet{wang2019factor} and \citet{Yu2021Projected}, we subtract the
monthly market excess returns and standardise the series one by one.

\begin{figure}[!htb]
\centering
\begin{minipage}{.45\textwidth}
		\centering
		\includegraphics[width=7.5cm, height=6cm]{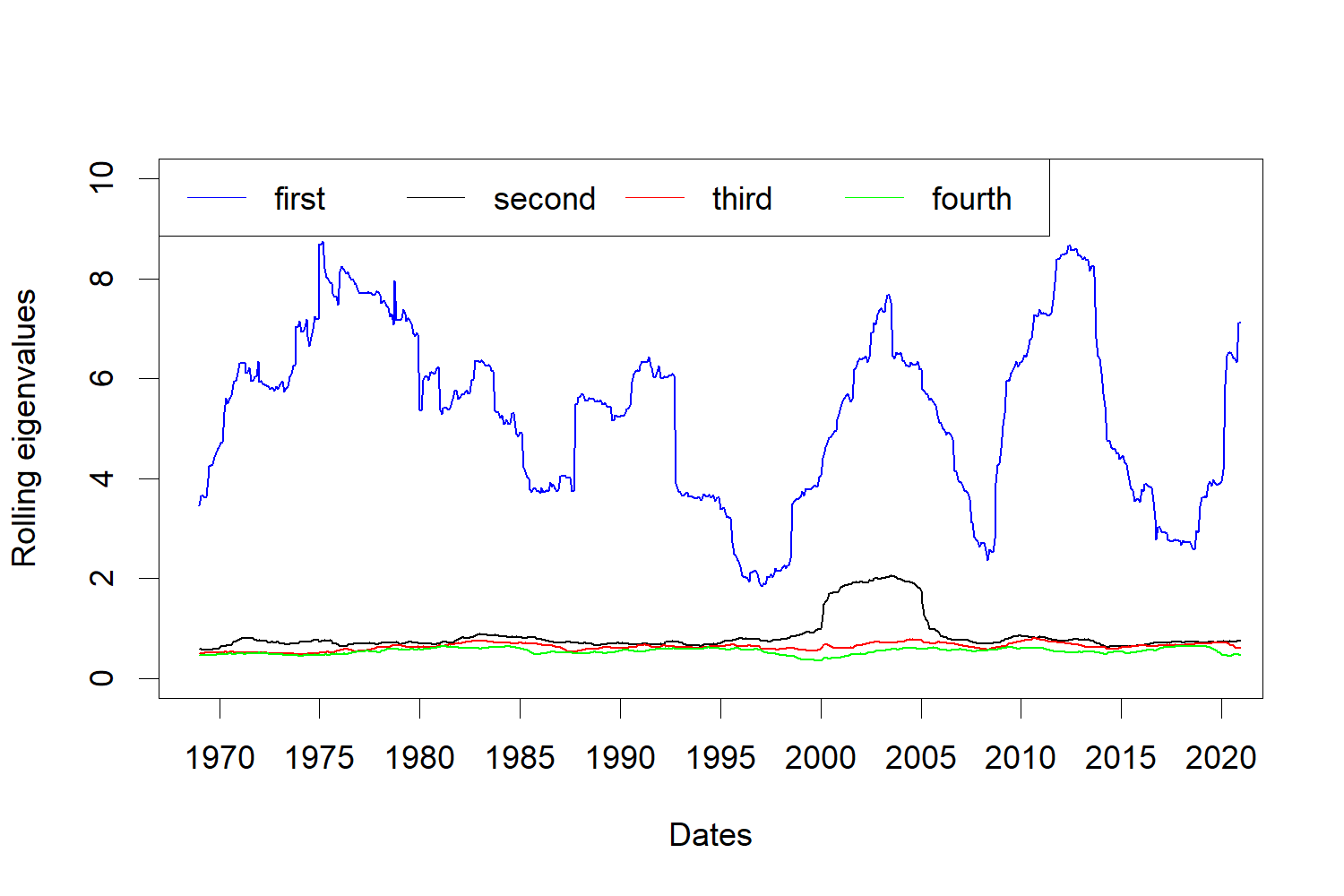}
	\end{minipage}
\begin{minipage}{0.45\textwidth}
		\centering
		\includegraphics[width=7.5cm, height=6cm]{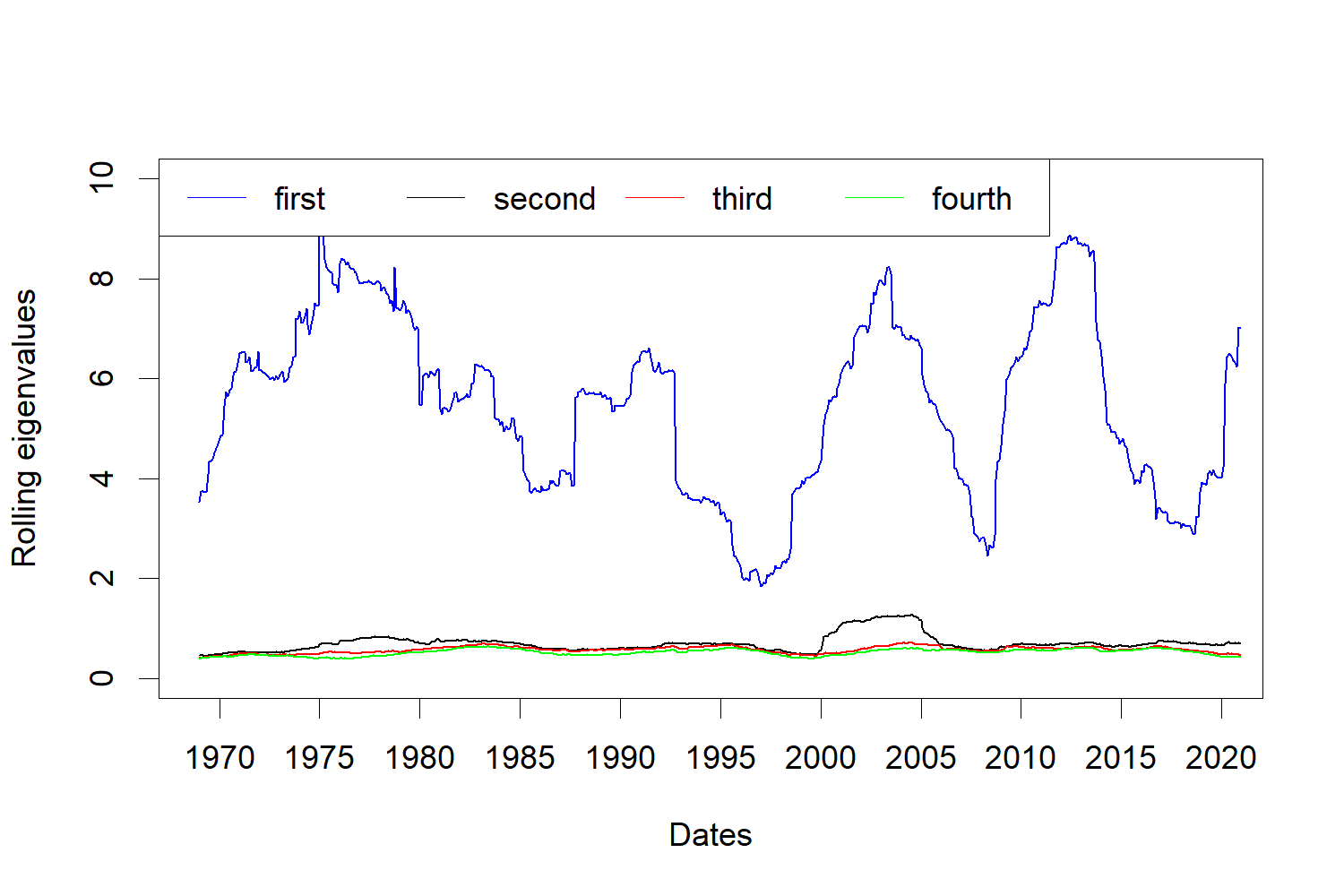}
	\end{minipage}
\caption{Eigenvalues of the rolling column-column and row-row sample
covariance matrices in the monitoring. Left: column-column (for $k_1$);
right: row-row (for $k_2$).}
\label{fig:eig}
\end{figure}

We implement our procedures at significance level $\alpha =0.05$, using, as
sampling period, $m=60$ ($5$ years) for all the monitoring schemes. The
first step of our analysis is to determine the number of row factor $k_{1}$.
It turns out the numbers of row and column factors are both found to be
equal to $1$, both when using the iterative algorithm in \cite%
{Yu2021Projected}, and the randomized testing procedure in \cite{hkyt}.
Therefore, we start with $k_{1}=k_{2}=1$ for the monitoring procedure.
Figure \ref{fig:eig} illustrates the four leading eigenvalues of the rolling
column-column and row-row projected sample covariance matrices. As shown in
Figure \ref{fig:eig}, the first largest eigenvalue is always much larger
than the remaining ones.

Our monitoring scheme detects, simultaneously, both changepoints in the
factor loading spaces, and in the number of factor. As a leading example, we
consider changes in the row factors, using the partial-sum testing
statistics with $\eta =0$. As $p_{1}$ and $p_{2}$ are much smaller than
those in our simulations, we use a larger $\delta =0.4$ when rescaling the
eigenvalues, to reduce the effects of the noises, but still use the
transformation function $g(x)=[\exp (x)-1]^{4}$. By way of robustness check,
we repeat the testing procedures $100$ times, and declare a change point
only when over $80\%$ of replications reject the corresponding null
hypothesis. The change point is viewed to be located at the median of values
from $100$ replications. We monitor until a changepoint is found, and,
whenever a change point is detected, e.g., at $\hat{\tau}_{j}$, until the
length of the remaining portion of the sample is smaller than $m$.

In the first step, we monitor the second largest eigenvalue of the rolling
column-column projected sample covariance matrix from $t=m+1$ to the end of
the sequence. The testing procedure outputs a change point in March 2003. We
also monitor the first eigenvalue, using the procedure suggested in Section %
\ref{further} to detect other types of change. In this case, we set $\tilde{g%
}(x)=1/g(x)$, finding no evidence of changepoints. Combining the results
illustrated by Figure \ref{fig:eig}, it is apparent that a new factor
emerges at $\tau _{1}=\text{\textquotedblleft March 2003\textquotedblright ,
with the second }$largest eigenvalue increasing significantly during this
period. Subsequently, we restart the monitoring process from $t=\tau _{1}+1$%
, using $k_{1}=2$, monitoring the third and second largest eigenvalues.
Interestingly, using the procedure proposed Section \ref{further}, we find
evidence that the number of common factors drops to $1$ at $\tau _{2}$=
\textquotedblleft June 2008\textquotedblright. This date is highly
suggestive, as it points towards concluding that, during the last global
financial crisis, the number of common factors reduced. This finding can be
read in conjunction with the empirical analysis in \citet{mst}, where the
number of common factors (albeit in a vector factor model with a threshold
structure) is found to decrease during downside market regimes; as the
authors put it, \textquotedblleft diversification disappears when needed
most\textquotedblright.

The change point locations for the row factors are plotted in the left panel
of Figure \ref{fig:change}, together with the maximum, median, and minimum
returns series after standardization. The red vertical line indicates the
time point when a new factor occurs, while the blue one indicates the time
point when a factor disappears. The estimated locations are very close to
the starting and ending points when the second largest eigenvalue increases
and decreases in Figure \ref{fig:eig}. Also, as expected, the return series
undergo a significant drop at the end of this period, say, the year 2008.

\begin{figure}[!htb]
\centering
\begin{minipage}{.45\textwidth}
		\centering
		\includegraphics[width=7.5cm, height=6cm]{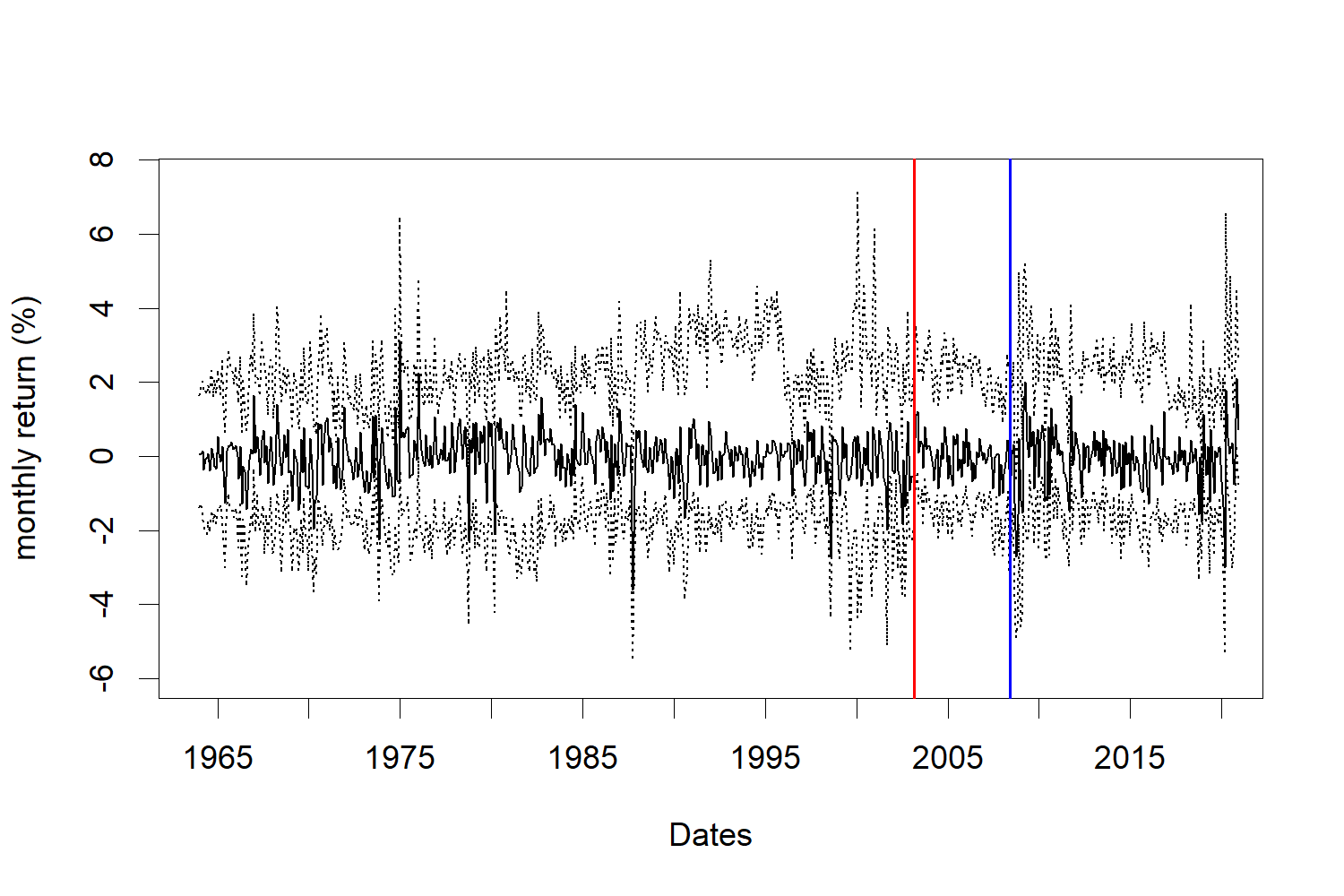}
	\end{minipage}
\begin{minipage}{0.45\textwidth}
		\centering
		\includegraphics[width=7.5cm, height=6cm]{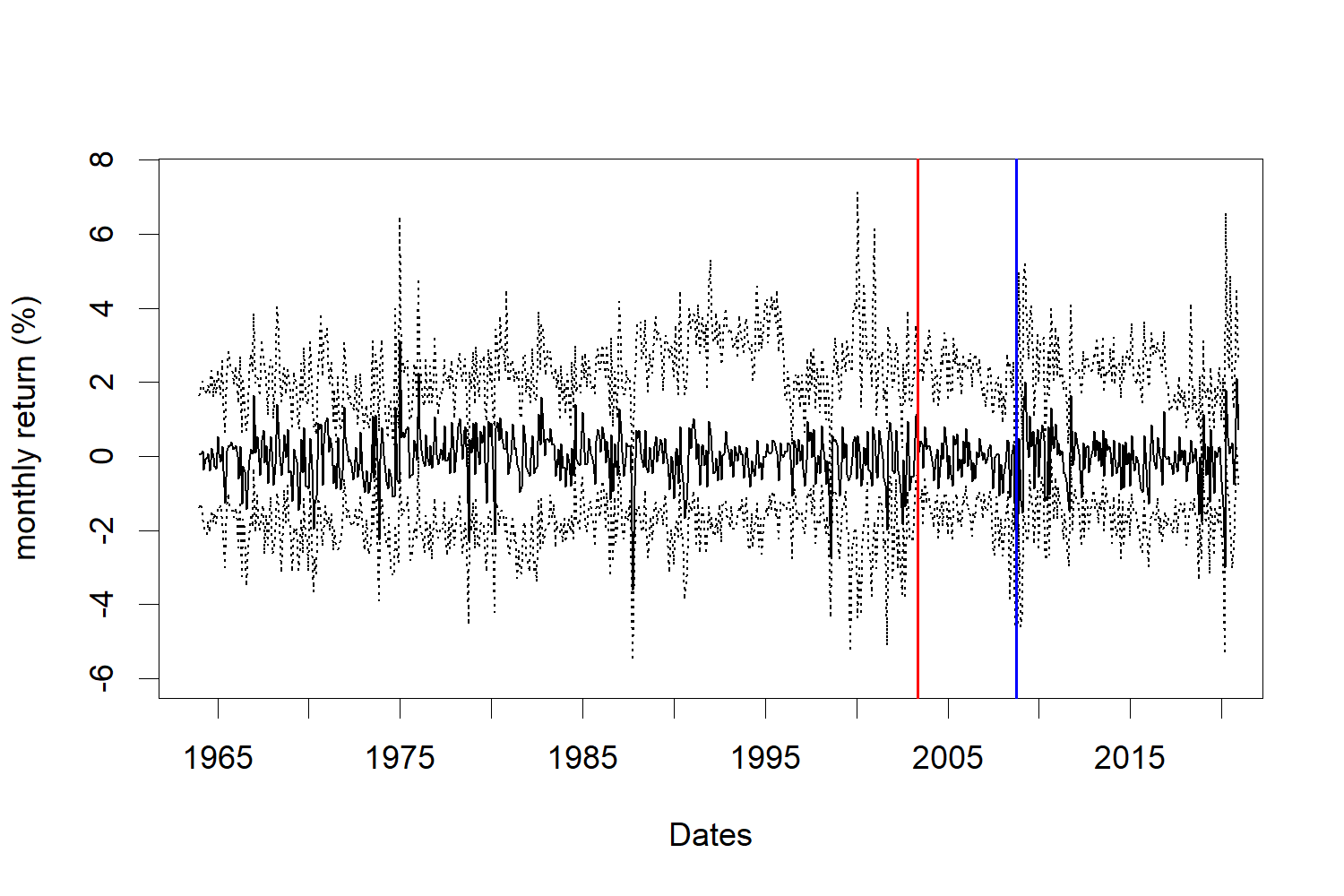}
	\end{minipage}
\caption{Illustration of the change points for the Fama-French 100 portfolio
series by partial-sum monitoring procedure with $\protect\eta=0$. Left: row
factors; right: column factors. In each panel, the dash lines are the
maximum and minimum return series, while the real line is the median return.
The red and blue vertical lines are the change point locations when the
number of factors increases and decreases, respectively. }
\label{fig:change}
\end{figure}

As far as the column factors are concerned, we can run our procedure in
parallel to detect any changes. However, by Figure \ref{fig:eig}, the second
largest eigenvalue in the right panel grows, but not as significantly as in
the left panel. Indeed, the monitoring procedure fails to detect any change
point for column factors if we still use the same $g(x)$ and $\tilde{g}(x)$
as for the row factors, and therefore there seems to be no sufficient
evidence to assert that the column loadings change during the monitoring
period. Again for robustness, we slightly modified the function $g(x)$, to $%
g(x)=[\exp (x)-1]^{2}$ and $\tilde{g}(x)=1/g(x)$; in this case, we point out
that our monitoring procedures do find evidence of two changepoints, as
plotted in the right panel of Figure \ref{fig:change}, whose locations are
very close to those found for the row factors. This suggests that, if
changepoints in the column factor spaces do exist, they are less evident and
\textquotedblleft pervasive\textquotedblright\ than in the case of the row
factor space.

\begin{table*}[hbtp]
\caption{Locations of change points detected by different monitoring
procedures. The numbers in bracket indicates the type of change for the
factor numbers, where ``1'' indicates new factor emerges; ``-1" indicates a
factor disappears while ``NA'' indicates no change.}
\label{tab:change ff}
\begin{center}
{\small \ \addtolength{\tabcolsep}{0pt} \renewcommand{\arraystretch}{1.2}
\scalebox{0.75}{ 	
				\begin{tabular*}{21cm}{ccccccccccccccc}
					\toprule[1.2pt]
					&\multirow{2}{*}{Change points}&\multicolumn{5}{c}{Partia-sum}&Worst-case\\\cmidrule(lr){3-7}
					&&$\eta=0$&$\eta=0.25$&$\eta=0.5$&$\eta=0.65$&$\eta=0.75$&&\\\midrule[1.2pt]
					\multirow{2}{*}{Row factors}&First&March 2003(1)&March 2002(1)&December 2003(1)&NA&NA&NA\\
					&Second&June 2008(-1)&April 2007(-1)& January 2009(-1) &NA&NA&NA\\
					\bottomrule[1.2pt]	
					\multirow{2}{*}{Column factors}&First&May 2003(1)&December 2002(1)&NA&NA&NA&NA\\
					&Second&October 2008(-1)&February 2008(-1)& NA &NA&NA&NA\\
					\bottomrule[1.2pt]		
		\end{tabular*}} }
\end{center}
\end{table*}

We also repeated the above procedure using partial-sum statistics with
different $\eta $, and the worst-case scenario statistic. Whilst, as can be
expected, changepoints are not always detected by all procedures, it turns
out that at most two change points are detected with all our monitoring
schemes; we summarize our findings in Table \ref{tab:change ff}, from which
we can conclude the results are, broadly, not sensitive to the monitoring
procedures.

\subsection{Multinational macroeconomic indices\label{mmi}}

In this application, we study changepoint detection with macro data. We use
the same data set as in \cite{Yu2021Projected}, which contains quarterly
observations for $10$ macroeconomic indices over $8$ countries from 1988-Q1
to 2020-Q2.\footnote{%
The countries are the United States, the United Kingdom, Canada,
France,Germany, Norway, Australia and New Zealand. The macroeconomic indices
are from 4 groups, namely Consumer Price Index (CPI), Interest Rate (IR),
Production (PRO) and International Trade (IT). The data can be freely
downloaded from OECD data library.} Similar data sets have also been studied
in \cite{liu2019helping} and \cite{wang2019factor}, which involves
macroeconomic indices from more countries. We refer to \cite{Yu2021Projected}
for further details on the dataset, and the preprocessing steps.

The first step is to determine the numbers of row and column factors. In
Figure \ref{fig:eig2}, we also plot the leading five sample eigenvalues in
the rolling monitoring process as in the first real example. Using the
randomized testing procedures in \cite{hkyt} suggests $k_{1}=1$, $k_{2}=3$
or $k_{1}=2$, $k_{2}=4$, which essentially indicates that one row and one
column factor may be weaker than the others. Combining this information with
the eigenvalue gaps shown in Figure \ref{fig:eig2}, we take $k_{1}=1$ and $%
k_{2}=3$.

\begin{figure}[!htb]
\centering
\begin{minipage}{.45\textwidth}
		\centering
		\includegraphics[width=7.5cm, height=6cm]{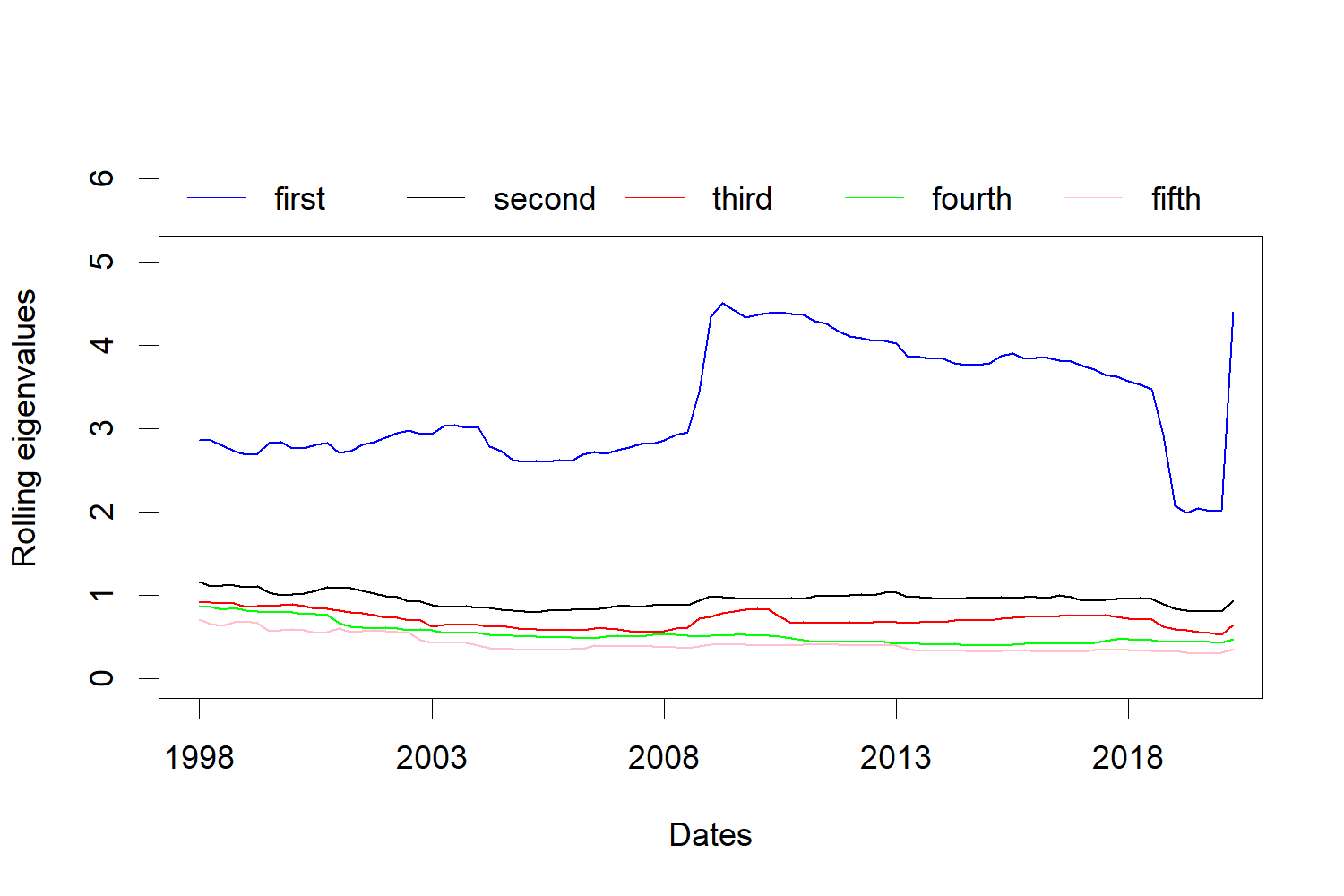}
	\end{minipage}
\begin{minipage}{0.45\textwidth}
		\centering
		\includegraphics[width=7.5cm, height=6cm]{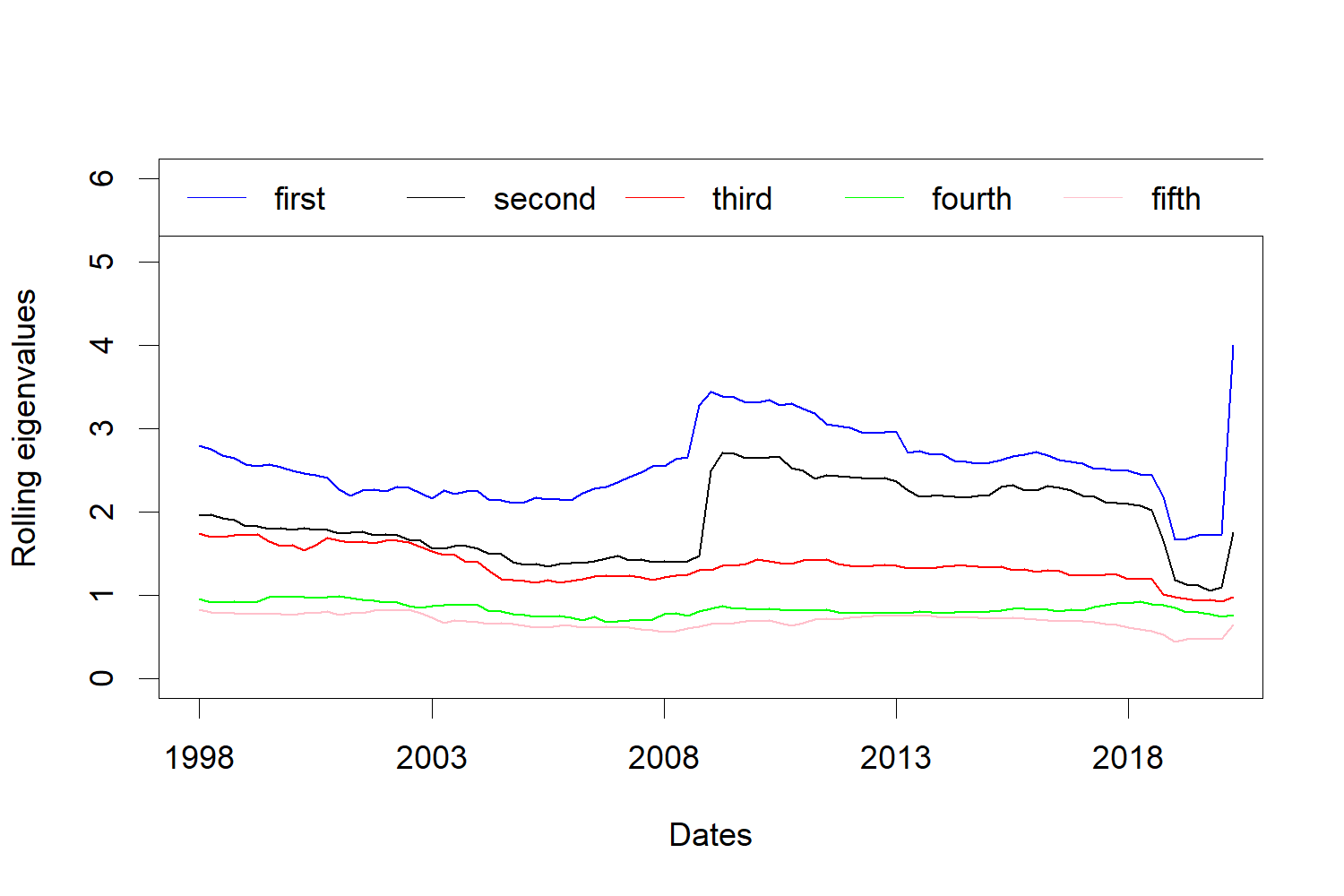}
	\end{minipage}
\caption{Eigenvalues of the rolling column-column and row-row sample
covariance matrices in the monitoring for macroeconomic indices. Left:
column-column (for $k_1$); right: row-row (for $k_2$).}
\label{fig:eig2}
\end{figure}

In the second step, we use same approach as in the previous section to
detect changepoints, i.e. we detect the change of factor loading spaces and
the change of factor numbers simultaneously. We use $m=40$ ($10$ years),
while setting the other tuning parameters the same as in the previous
section, considering that $p_{1}$ and $p_{2}$ are of comparable size. For
the partial-sum statistics with $\eta =0$, it turns out there are no change
points for the row factors, while one column factor disappears in 2004-Q3.
We plot the results in Figure \ref{fig:change macro}, together with the
median CPI series and median IR series. Interestingly, the location of the
change point closely matches the changepoint detected in the previous
example. On the other hand, by Figure \ref{fig:eig2}, the third column
eigenvalue (right panel) starts to decrease in the year 2004, and becomes
relatively smaller than the leading two eigenvalues after the year 2008.
Moreover, the CPI series becomes unstable after the year 2002, while the
variance of IR series starts to decrease at the same time. Factoring in the
possible delay in detection, this suggests a possible explanation for the
2004-Q3 break.

\begin{figure}[!htb]
\centering
\includegraphics[width=15cm, height=10cm]{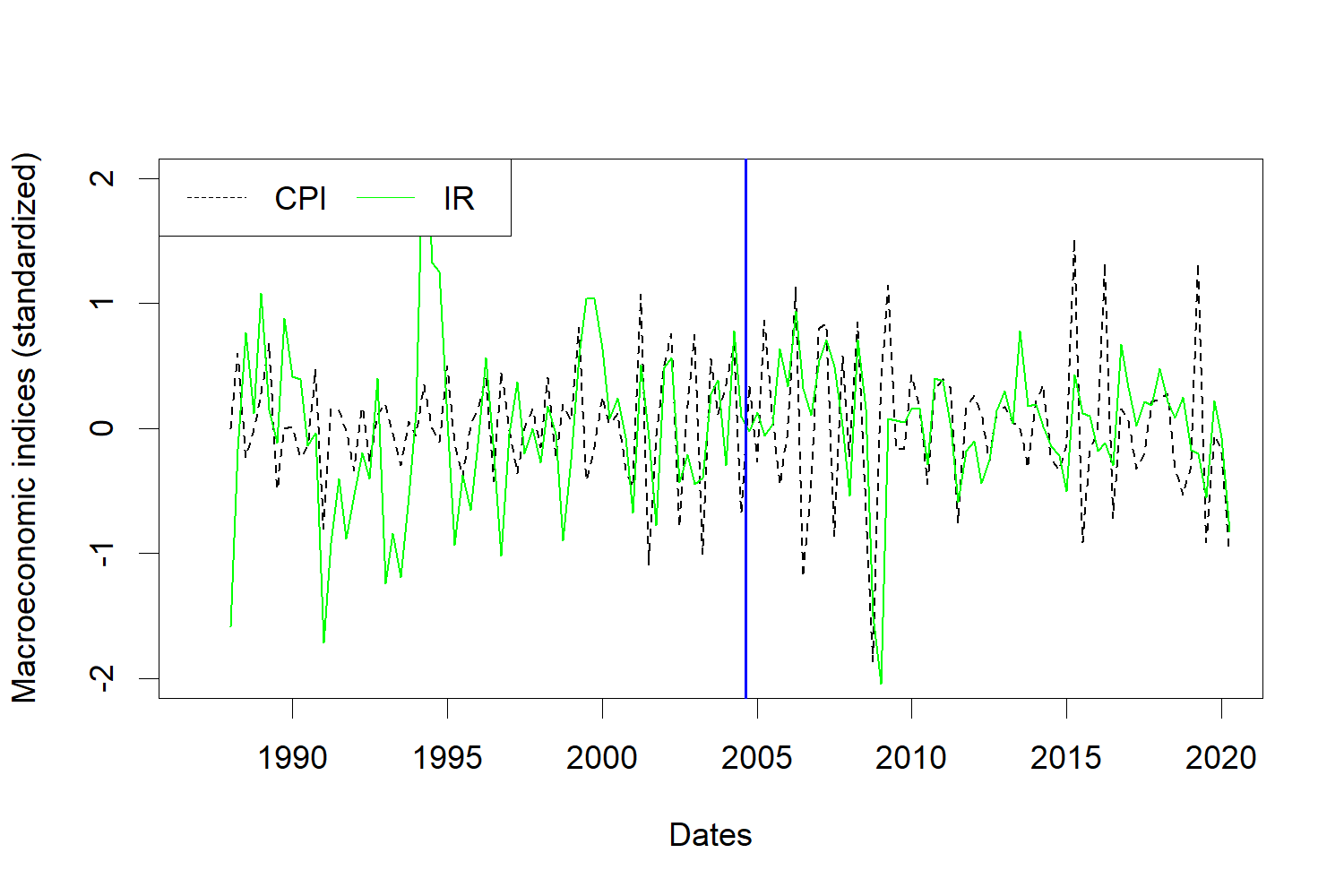}
\caption{Illustration of the change point of column factors for
macroeconomic indices. Black and green lines indicate the median CPI and
interest rate series. Blue vertical line indicates the change point when a
column factor disappears. }
\label{fig:change macro}
\end{figure}

For other choices of $\eta $ and the worst-case statistic, we summarize the
result in Table \ref{tab:change macro}. Interestingly, the changepoints we
detect in this example are all roughly close to the two change points
detected in the previous example. We also remark that in this macroeconomic
example, the length of series is much smaller than that in the last one.
Actually, for the testing statistics with $\eta >0.5$, the effective sample
size in the monitoring process is $\ln (T_{m})\leq \ln 90$, which is very
small. Hence, in this example we prefer to trust the results from testing
procedures with smaller $\eta $ or the worst-case testing procedure.

\begin{table*}[hbtp]
\caption{Estimated change points locations by different monitoring
procedures. The numbers -1 in brackets indicates a factor disappears. ``NA''
means no change point.}
\label{tab:change macro}
\begin{center}
{\small \ \addtolength{\tabcolsep}{0pt} \renewcommand{\arraystretch}{1.2}
\scalebox{0.75}{ 	
				\begin{tabular*}{21cm}{ccccccccccccccc}
					\toprule[1.2pt]
					&\multirow{2}{*}{Change points}&\multicolumn{5}{c}{Partia-sum}&Worst-case\\\cmidrule(lr){3-7}
					&&$\eta=0$&$\eta=0.25$&$\eta=0.5$&$\eta=0.65$&$\eta=0.75$&&\\\midrule[1.2pt]
					Row factor&First&NA&NA&NA&NA&NA&NA\\
					\bottomrule[1.2pt]	
					Column factor&First&2004-Q3(-1)&2003-Q1(-1)&2000-Q3(-1)&1999-Q2(-1)&1998-Q3(-1)&2008-Q3(-1)\\
					\bottomrule[1.2pt]		
		\end{tabular*}} }
\end{center}
\end{table*}

\section{Conclusions\label{conclusion}}

In this paper, we have proposed several schemes for the online detection of
changes in the latent factor structures of a two-way, matrix factor model.
Our approach is based on noting that many instances of changepoint can be
represented as a change in the dimension of the factor spaces. Hence, in
order to detect changes, we use the eigenvalues of the projected second
moment matrices, which diverge with the cross-sectional sample or not
according to the number of common factors. Having a spiked behaviour in an
eigenvalue which was bounded during a training period, or observing the
vanishing of such a spiked behaviour, point towards the presence of a break.
Since we do not know the limiting distribution of the estimated eigenvalues,
which is likely to be very challenging to derive especially in the case
where the eigenvalue is not spiked, we randomise the sequence of the
estimated eigenvalue by perturbing it with an \textit{i.i.d.}, standard
normal sequence. Thence, we are able to propose two families of sequential
procedures, one more \textquotedblleft classical\textquotedblright\ and
based on the fluctuations of partial sums, and another one, completely
novel, based on the extreme value behaviour of the perturbed sequence of
estimated eigenvalues. Our approach has several distinctive advantages: it
is easy to implement, it requires much less tuning than competing approaches
(see e.g. \citealp{BT1}), and it works very well in simulations, offering
good size control and fast changepoint detection. Indeed, whilst approaches
based on partial sums work very well, the methodology based on the extreme
value, worst-case scenario delivers an even superior performance, thus
suggesting that our methodology could be very promising even in other, very
different contexts.

{\small
\bibliographystyle{chicago}
\bibliography{LTbiblio}
}

\setcounter{equation}{0} \setcounter{lemma}{0} \renewcommand{\thelemma}{A.%
\arabic{lemma}} \renewcommand{\theequation}{A.\arabic{equation}} \appendix

\setcounter{section}{0} \setcounter{subsection}{-1}

\section{Further assumptions\label{assumptions}}

The following assumptions are borrowed from the paper by \cite%
{Yu2021Projected}, to which we refer for detailed explanations and
discussions.

\begin{assumpB}
\label{factors} (i) (a) $E(F_{t})=0$, and (b) $E\Vert F_{t}\Vert
^{4+\epsilon }\leq c_{0}$, for some $\epsilon >0$; (ii)
\begin{equation}
\frac{1}{T}\sum_{t=1}^{T}F_{t}F_{t}^{\prime }\overset{a.s.}{\rightarrow }{%
\Sigma}_{1}\text{ and }\frac{1}{T}\sum_{t=1}^{T}F_{t}^{\prime }F_{t}\overset{%
a.s.}{\rightarrow }\Sigma_{2},  \label{equ:covariance}
\end{equation}%
where $\Sigma_{i}$ is a $k_{i}\times k_{i}$ positive definite matrix with
distinct eigenvalues and spectral decomposition $\Sigma_{i}=\Gamma_{i}%
\Lambda_{i}\Gamma_{i}^{\prime}$, $i=1,2$. The factor numbers $k_{1}$ and $%
k_{2} $ are fixed as $\min \{T,p_{1},p_{2}\}\rightarrow \infty $; (iii) it
holds that, for all $1\leq h_{1},l_{1}\leq k_{1}$ and $1\leq h_{2},l_{2}\leq
k_{2}$%
\[
E\max_{1\leq \widetilde{t}\leq T}\left( \sum_{t=1}^{\widetilde{t}}\left(
F_{h_{1}h_{2},t}F_{l_{1}l_{2},t}-E\left(
F_{h_{1}h_{2},t}F_{l_{1}l_{2},t}\right) \right) \right) ^{2}\leq c_{0}T.
\]
\end{assumpB}

\begin{assumpB}
\label{loading} (i) $\Vert R\Vert_{\max}\leq c_{0}$, and $\Vert
C\Vert_{\max} \leq c_{1}$; (ii) as $\min \{p_{1},p_{2}\}\rightarrow \infty $%
, $\Vert p_{1}^{-1}R^{\top }R-I_{k_{1}}\Vert \rightarrow 0$ and $\Vert
p_{2}^{-1}C^{\top }C-I_{k_{2}}\Vert \rightarrow 0$.
\end{assumpB}

Assumptions \ref{factors} and \ref{loading} are standard in large factor
models, and we refer, for example, to \citet{fan2021}. In Assumption \ref%
{factors}\textit{(i)}(b), note the (mild) strengthening of the customarily
assumed fourth moment existence condition on $F_{t}$ - this is required in
order to prove our results, which rely on almost sure rates. Similarly, the
maximal inequality in part \textit{(iii)} of the assumption is usually not
considered in the literature, and it can be derived from more primitive
dependence assumptions: for example, it can be shown to hold under various
mixing conditions (see e.g. \citealp{rio1995maximal}; and %
\citealp{shao1995maximal}), and for the very general class of decomposable
Bernoulli shifts (see e.g. \citealp{berkeshormann}; \citealp{linliu}; and %
\citealp{BT2}). Finally, we point out that, according to Assumption \ref%
{loading}, the common factors are pervasive. Extensions to the case of
\textquotedblleft weak\textquotedblright\ factors go beyond the scope of
this paper, but are in principle possible.

\begin{assumpB}
\label{idiosyncratic} (i) (a) $E(e_{ij,t})=0$, and (b) $E(e_{ij,t}^{8})\leq
c_{0}$; (ii) for all $1\leq t\leq T$, $1\leq i\leq p_{1}$ and $1\leq j\leq
p_{2}$,
\[
(\text{a}).\sum_{s=1}^{T}\sum_{l=1}^{p_{1}}%
\sum_{h=1}^{p_{2}}|E(e_{ij,t}e_{lh,s})|\leq c_{0},\quad (\text{b}%
).\sum_{l=1}^{p_{1}}\sum_{h=1}^{p_{2}}|{E}(e_{lj,t}e_{ih,t})|\leq c_{0};
\]%
(iii) for all $1\leq t\leq T$, $1\leq i\leq p_{1}$ and $1\leq j\leq p_{2}$,%
\[
\begin{array}{cl}
\left( \text{a}\right) . & \sum_{s=1}^{T}\sum_{l_{2}=1}^{p_{1}}%
\sum_{h=1}^{p_{2}}\left\vert
Cov(e_{ij,t}e_{l_{1}j,t},e_{ih,s}e_{l_{2}h,s})\right\vert \leq c_{0}, \\
& \sum_{s=1}^{T}\sum_{l=1}^{p_{1}}\sum_{h_{2}=1}^{p_{2}}\left\vert
Cov(e_{ij,t}e_{ih_1,t},e_{lj,s}e_{lh_2,s})\right\vert \leq c_{0}, \\
& \sum_{s=1}^{T}\sum_{l=1}^{p_{1}} \sum_{h=1}^{p_{2}}\left\vert
Cov(e_{ij,t}^2,e_{lh,s}^2)\right\vert \leq c_{0}, \\
\left( \text{b}\right) . & \sum_{s=1}^{T}\sum_{l_{2}=1}^{p_{1}}%
\sum_{h_{2}=1}^{p_{2}}\left\vert
Cov(e_{ij,t}e_{l_{1}h_{1},t},e_{ij,s}e_{l_{2}h_{2},s})+Cov(e_{l_{1}j,t}e_{ih_{1},t},e_{l_{2}j,s}e_{ih_{2},s})\right\vert \leq c_{0},%
\end{array}%
\]%
{\small (iv) it holds that }$\lambda _{\min }\left[ E\left( \frac{1}{p_{2}T}%
\sum_{t=1}^{T}E_{t}E_{t}^{\prime }\right) \right] >0$ and $\lambda _{\min }%
\left[ E\left( \frac{1}{p_{1}T} \sum_{t=1}^{T}E_{t}^{\prime}E_{t}\right) %
\right] >0$.
\end{assumpB}

Assumption \ref{idiosyncratic} ensures the (cross-sectional and time series)
summability of the idiosyncratic terms $E_{t}$. The assumption requires weak
dependence in both the space and time domains; in principle, it is possible
to show that this assumption is satisfied for many dependence assumptions
and data generating processes. The assumption can be read in conjunction
with the paper by \citet{wang2019factor}, where $E_{t}$ is assumed to be
white noise, which can be viewed as overly restrictive.

\begin{assumpB}
\label{depFE} (i) For any deterministic vectors ${v}$ and ${w}$ satisfying $%
\Vert {v}\Vert =1$ and $\Vert {w}\Vert =1$ with suitable dimensions,
\[
{E}\bigg\|\frac{1}{\sqrt{T}}\sum_{t=1}^{T}F_{t}({v}^{\prime}E_{t}{w})\bigg\|%
^{2}\leq c_{0};
\]%
(ii) for all $1\leq i,l_{1}\leq p_{1}$ and $1\leq j,h_{1}\leq p_{2}$,{\small
\[
\begin{split}
& (\text{a}).\Big\|\sum_{h=1}^{p_{2}}{E}(\bar{\zeta}_{ij}\otimes \bar{\zeta}%
_{ih})\Big\|_{\max }\leq c_{0},\quad \Big\|\sum_{l=1}^{p_{1}}{E}(\bar{\zeta}%
_{ij}\otimes \bar{\zeta}_{lj})\Big\|_{\max }\leq c_{0}, \\
& (\text{b}).\Big\|\sum_{l=1}^{p_{1}}\sum_{h_{2}=1}^{p_{2}}Cov(\bar{\zeta}%
_{ij}\otimes \bar{\zeta}_{ih_{1}},\bar{\zeta}_{lj}\otimes \bar{\zeta}%
_{lh_{2}})\Big\|_{\max }\leq c_{0},\Big\|\sum_{l_{2}=1}^{p_{1}}%
\sum_{h=1}^{p_{2}}Cov(\bar{\zeta}_{ij}\otimes \bar{\zeta}_{l_{1}j},\bar{\zeta%
}_{ih}\otimes \bar{\zeta}_{l_{2}h})\Big\|_{\max }\leq c_{0},
\end{split}%
\]%
where $\bar{\zeta}_{ij}=\text{Vec}(\sum_{t=1}^{T}F_{t}e_{ij,t}/\sqrt{T})$. }
\end{assumpB}

According to Assumption \ref{depFE}, the common factors $F_{t}$ and the
errors $E_{t}$ can be weakly correlated. The assumption holds under the more
restrictive case that $\{F_{t}\}$ and $\{E_{t}\}$ are two mutually
independent groups.

\setcounter{equation}{0} \setcounter{lemma}{0} \renewcommand{\thelemma}{B.%
\arabic{lemma}} \renewcommand{\theequation}{B.\arabic{equation}}

\subsection{Technical lemmas and proofs\label{proofs}}

\begin{lemma}
\label{theorem:tildeM1 copy(1)} We assume that Assumptions \ref{factors}-\ref%
{depFE} are satisfied. Then it holds that there exist a triplet of random
variables $\left( p_{1,0},p_{2,0},m_{0}\right) $ such that

\begin{description}
\item[(i)] under (\ref{b1})
\[
\widehat{\lambda }_{k_{1}+1,\tau }\left\{
\begin{array}{ll}
\leq c_{0} & \text{for }\tau \leq t^{\ast } \\
\geq c_{1}\min \left( \frac{\tau -t^{\ast }}{m},\frac{m+t^{\ast }-\tau }{m}%
\right) p_{1} & \text{for }t^{\ast }<\tau <m+t^{\ast } \\
\leq c_{0} & \text{for }\tau \geq m+t^{\ast }%
\end{array}%
\right. ;
\]

\item[(ii)] under (\ref{b2})%
\[
\widehat{\lambda }_{k_{1}+1,\tau }\left\{
\begin{array}{ll}
\leq c_{0} & \text{for }\tau \leq t^{\ast } \\
\geq c_{1}\frac{\tau -t^{\ast }}{m}p_{1} & \text{for }t^{\ast }<\tau
<m+t^{\ast } \\
\geq c_{0}p_{1} & \text{for }\tau \geq m+t^{\ast }%
\end{array}%
\right. .
\]
\end{description}

\begin{proof}
The proof of the lemma is based on very similar passages as the proof of
Lemma 1 in \citet{BT1}, and thus we report only the main passages. We begin
with part \textit{(ii)} of the lemma, focusing on the case $t^{\ast }<\tau
<m+t^{\ast }$. Note that%
\begin{eqnarray*}
\widehat{\lambda }_{k_{1}+1,\tau } &=&\widehat{\lambda }_{k_{1}+1}\left(
\frac{1}{m}\sum_{t=\tau +1}^{m+t^{\ast }}Y_{t}Y_{t}^{\prime }+\frac{1}{m}%
\sum_{t=m+t^{\ast }}^{m+\tau }Y_{t}Y_{t}^{\prime }\right)  \\
&\geq &\widehat{\lambda }_{\min }\left( \frac{1}{m}\sum_{t=\tau
+1}^{m+t^{\ast }}Y_{t}Y_{t}^{\prime }\right) +\widehat{\lambda }%
_{k_{1}+1}\left( \frac{1}{m}\sum_{t=m+t^{\ast }}^{m+\tau }Y_{t}Y_{t}^{\prime
}\right)  \\
&\geq &\widehat{\lambda }_{k_{1}+1}\left( \frac{1}{m}\sum_{t=m+t^{\ast
}}^{m+\tau }Y_{t}Y_{t}^{\prime }\right) ,
\end{eqnarray*}%
by Weyl's inequality. Also%
\begin{equation*}
\widehat{\lambda }_{k_{1}+1}\left( \frac{1}{m}\sum_{t=m+t^{\ast }}^{m+\tau
}Y_{t}Y_{t}^{\prime }\right) =\frac{\tau -t^{\ast }}{m}\widehat{\lambda }%
_{k_{1}+1}\left( \frac{1}{m+\tau -t^{\ast }}\sum_{t=m+t^{\ast }}^{m+\tau
}Y_{t}Y_{t}^{\prime }\right) ;
\end{equation*}%
the fact that%
\begin{equation*}
\widehat{\lambda }_{k_{1}+1}\left( \frac{1}{\tau -t^{\ast }}%
\sum_{t=m+t^{\ast }}^{m+\tau }Y_{t}Y_{t}^{\prime }\right) \geq c_{0}p_{1},
\end{equation*}%
is an immediate consequence of Lemma \ref{theorem:tildeM1}. The rest of the
proof - for the cases $\tau \leq t^{\ast }$ and $\tau \geq m+t^{\ast }$\ -
follows immediately. We now turn to part (i) of the Lemma, and consider the
case where all loadings change for simplicity - i.e. $R_{0}$ is empty. In
the interval $t^{\ast }<\tau <m+t^{\ast }$ we can write%
\begin{equation*}
X_{t}=\left[
\begin{array}{cc}
R_{1} & R_{2}%
\end{array}%
\right] \left[
\begin{array}{cc}
F_{t}I\left( t<t^{\ast }\right)  & 0 \\
0 & F_{t}I\left( t\geq t^{\ast }\right)
\end{array}%
\right] \left[
\begin{array}{c}
C^{\prime } \\
C^{\prime }%
\end{array}%
\right] +E_{t}.
\end{equation*}%
The second moment matrix has a component given by%
\begin{equation*}
\left[
\begin{array}{cc}
\frac{1}{m}\sum_{t=\tau +1}^{m+t^{\ast }}R_{1}F_{t}I\left( t<t^{\ast
}\right) C^{\prime }\widehat{C}\widehat{C}^{\prime }CF_{t}^{\prime
}R_{1}^{\prime } & 0 \\
0 & \frac{1}{m}\sum_{t=m+t^{\ast }}^{m+\tau }R_{2}F_{t}I\left( t\geq t^{\ast
}\right) C^{\prime }\widehat{C}\widehat{C}^{\prime }CF_{t}^{\prime
}R_{2}^{\prime }%
\end{array}%
\right]
\end{equation*}%
The spectrum of this matrix is the union of the spectra of the two blocks,
and therefore, using Lemma \ref{theorem:tildeM1}, it follows that%
\begin{eqnarray*}
&&\widehat{\lambda }_{k_{1}+1}\left( \frac{1}{m}\sum_{t=m+t^{\ast }}^{m+\tau
}Y_{t}Y_{t}^{\prime }\right)  \\
&\geq &c_{0}\min \left\{ \widehat{\lambda }_{k_{1}}\left( \frac{1}{m}%
\sum_{t=\tau +1}^{m+t^{\ast }}R_{1}F_{t}I\left( t<t^{\ast }\right) C^{\prime
}\widehat{C}\widehat{C}^{\prime }CF_{t}^{\prime }R_{1}^{\prime }\right) ,%
\widehat{\lambda }_{k_{1}}\left( \frac{1}{m}\sum_{t=m+t^{\ast }}^{m+\tau
}R_{2}F_{t}I\left( t\geq t^{\ast }\right) C^{\prime }\widehat{C}\widehat{C}%
^{\prime }CF_{t}^{\prime }R_{2}^{\prime }\right) \right\}  \\
&\geq &c_{0}\min \left\{ \frac{m+t^{\ast }-\tau }{m},\frac{\tau -t^{\ast }}{m%
}\right\} ,
\end{eqnarray*}%
whence the desired result follows.
\end{proof}
\end{lemma}

\begin{proof}[Proof of Theorem \ref{lemma:function}]
We begin by showing (\ref{weighted}), where $\eta <1/2$. Define%
\begin{equation*}
S_{\tau }^{\ast }=\sum_{j=1}^{\tau }z_{j},
\end{equation*}%
and recall that $\left\{ z_{\tau },1\leq \tau \leq T_{m}\right\} $ is
\textit{i.i.d.} $N\left( 0,1\right) $. Since $\psi _{\tau }\geq 0$ for all $%
\tau $, it holds that%
\begin{equation}
\max_{1\leq \tau \leq T_{m}}\frac{\left\vert S_{\tau }-S_{\tau }^{\ast
}\right\vert }{\tau ^{\eta }}=\max_{1\leq \tau \leq T_{m}}\frac{%
\sum_{j=1}^{\tau }\psi _{j}}{\tau ^{\eta }}=T_{m}^{1-\eta }g\left(
p_{1}^{-\delta }\frac{p_{1}}{T^{1/2}}\left( \ln p_{1}\right) ^{1+\epsilon
}\left( \ln p_{2}\right) ^{1+\epsilon }\left( \ln T\right) ^{1+\epsilon
}\right) =o_{P^{\ast }}\left( 1\right) ,  \label{max-functional}
\end{equation}%
due to (\ref{restriction}). By the KMT approximation (see \citealp{KMT1},
and \citealp{KMT2}), there exists a sequence of standard Wiener process $%
W_{T_{m}}\left( x\right) $, $0<x<\infty $, such that%
\begin{equation}
\max_{1\leq \tau \leq T_{m}}\frac{1}{\ln \tau }\left\vert S_{\tau }^{\ast
}-W_{T_{m}}\left( \tau \right) \right\vert =O_{P}\left( 1\right) .
\label{kmt}
\end{equation}%
Then, using (\ref{kmt}), it follows that%
\begin{eqnarray}
&&T_{m}^{\eta -1/2}\max_{1\leq \tau \leq Tm}\tau ^{-\eta }\left\vert S_{\tau
}^{\ast }-W_{T_{m}}\left( \tau \right) \right\vert  \label{lajos} \\
&\leq &T_{m}^{\eta -1/2}\max_{1\leq \tau \leq Tm}\frac{1}{\ln \tau }%
\left\vert S_{\tau }^{\ast }-W_{T_{m}}\left( \tau \right) \right\vert
\max_{1\leq \tau \leq Tm}\tau ^{-\eta }\ln \tau  \notag \\
&=&o_{P}\left( 1\right) .  \notag
\end{eqnarray}%
Finally%
\begin{eqnarray}
&&T_{m}^{\eta -1/2}\max_{1\leq \tau \leq Tm}\tau ^{-\eta }\left\vert
W_{T_{m}}\left( \tau \right) \right\vert  \label{max-wiener} \\
&=&\max_{1\leq \tau \leq Tm}\left( \frac{\tau }{T_{m}}\right) ^{-\eta
}\left\vert W_{T_{m}}\left( \frac{\tau }{T_{m}}\right) \right\vert  \notag \\
&&\overset{D}{=}\max_{1/T_{m}\leq u\leq 1}\left( u\right) ^{-\eta
}\left\vert W\left( u\right) \right\vert  \notag \\
&&\overset{a.s.}{\rightarrow }\max_{0\leq u\leq 1}\left( u\right) ^{-\eta
}\left\vert W\left( u\right) \right\vert ,  \notag
\end{eqnarray}%
as $T_{m}\rightarrow \infty $, having used the scale transformation of the
Wiener process. Putting (\ref{max-functional}), (\ref{kmt}), (\ref{lajos})
and (\ref{max-wiener}) together, (\ref{weighted}) follows.

We now turn to showing (\ref{darling-erdos}). Note first that%
\begin{eqnarray*}
&&\left\vert \max_{1\leq \tau \leq T_{m}}\tau ^{-1/2}\left\vert S_{\tau
}\right\vert -\max_{1\leq \tau \leq T_{m}}\tau ^{-1/2}\left\vert S_{\tau
}^{\ast }\right\vert \right\vert \\
&\leq &\max_{1\leq \tau \leq T_{m}}\tau ^{-1/2}\sum_{j=1}^{\tau }\psi _{j} \\
&\leq &T_{m}^{-1/2}\sum_{j=1}^{T_{m}}\psi _{j}=O_{a.s.}\left(
T_{m}^{1/2}g\left( p_{1}^{-\delta }\frac{p_{1}}{T^{1/2}}\left( \ln
p_{1}\right) ^{1+\epsilon }\left( \ln p_{2}\right) ^{1+\epsilon }\left( \ln
T\right) ^{1+\epsilon }\right) \right) \\
&=&o_{a.s.}\left( T_{m}^{-1/2}\right) ,
\end{eqnarray*}%
having used (\ref{mon-eig-null}) and (\ref{restriction}). Thus it holds that
\begin{equation}
\alpha _{Tm}\left\vert \max_{1\leq \tau \leq T_{m}}\left\vert S_{\tau
}\right\vert -\max_{1\leq \tau \leq T_{m}}\left\vert S_{\tau }^{\ast
}\right\vert \right\vert -\beta _{Tm}\overset{a.s.}{\rightarrow }-\infty ;
\label{shorack}
\end{equation}%
this proves equation (2.3) in \citet{Shorack1979Extension}. Hence
\begin{equation*}
P^{\ast }\left( \alpha _{Tm}\max_{1\leq \tau \leq T_{m}}\tau
^{-1/2}\left\vert S_{\tau }\right\vert \leq x+\beta _{Tm}\right) =P^{\ast
}\left( \alpha _{Tm}\max_{1\leq \tau \leq T_{m}}\tau ^{-1/2}\left\vert
S_{\tau }^{\ast }\right\vert \leq x+\beta _{Tm}\right) +o_{P^{\ast }}\left(
1\right) ;
\end{equation*}%
recalling that $S_{\tau }^{\ast }$ is the partial sums process of a sequence
of \textit{i.i.d.} standard normals, the Darling-Erd\H{o}s theorem (%
\citealp{darling1956limit}) yields the desired result.

Finally, we show (\ref{renyi}). The proof is similar to the above, so we
omit passages when this does not cause confusion. Note that%
\begin{eqnarray*}
&&r_{T_{m}}^{\eta -1/2}\left\vert \max_{r_{T_{m}}\leq \tau \leq T_{m}}\tau
^{-\eta }\left\vert S_{\tau }\right\vert -\max_{r_{T_{m}}\leq \tau \leq
T_{m}}\tau ^{-\eta }\left\vert S_{\tau }^{\ast }\right\vert \right\vert \\
&\leq &r_{T_{m}}^{\eta -1/2}\max_{r_{T_{m}}\leq \tau \leq T_{m}}\tau ^{-\eta
}\sum_{j=1}^{\tau }\psi _{j},
\end{eqnarray*}%
which is bounded by%
\begin{equation*}
O_{a.s.}\left( r_{T_{m}}^{\eta -1/2}T_{m}^{1-\eta }g\left( p_{1}^{-\delta }%
\frac{p_{1}}{T^{1/2}}\left( \ln p_{1}\right) ^{1+\epsilon }\left( \ln
p_{2}\right) ^{1+\epsilon }\left( \ln T\right) ^{1+\epsilon }\right) \right)
,
\end{equation*}%
whenever $\frac{1}{2}<\eta <1$, and by%
\begin{equation*}
O_{a.s.}\left( r_{T_{m}}^{1/2}g\left( p_{1}^{-\delta }\frac{p_{1}}{T^{1/2}}%
\left( \ln p_{1}\right) ^{1+\epsilon }\left( \ln p_{2}\right) ^{1+\epsilon
}\left( \ln T\right) ^{1+\epsilon }\right) \right) ,
\end{equation*}%
when $\eta >1$. In both cases, it is easy to see that the bounds drift to
zero as $\min \left( p_{1},p_{2},T\right) \rightarrow \infty $. Also%
\begin{eqnarray*}
&&r_{T_{m}}^{\eta -1/2}\max_{r_{T_{m}}\leq \tau \leq T_{m}}\tau ^{-\eta
}\left\vert S_{\tau }^{\ast }-W_{T_{m}}\left( \tau \right) \right\vert \\
&\leq &r_{T_{m}}^{\eta -1/2}\max_{r_{T_{m}}\leq \tau \leq T_{m}}\tau ^{-\eta
}\frac{\left\vert S_{\tau }^{\ast }-W_{T_{m}}\left( \tau \right) \right\vert
}{\ln \tau }\ln \tau \\
&\leq &O_{P}\left( 1\right) r_{T_{m}}^{\eta -1/2}\max_{r_{T_{m}}\leq \tau
\leq T_{m}}\tau ^{-\eta }\ln \tau \\
&=&O_{P}\left( r_{T_{m}}^{-1/2}\ln r_{T_{m}}\right) =o_{P}\left( 1\right) ,
\end{eqnarray*}%
where $S_{\tau }^{\ast }$ and $W_{T_{m}}\left( \tau \right) $ are defined
above. Hence we need to study%
\begin{equation*}
r_{T_{m}}^{\eta -1/2}\max_{r_{T_{m}}\leq \tau \leq T_{m}}\tau ^{-\eta
}\left\vert W_{T_{m}}\left( \tau \right) \right\vert ;
\end{equation*}%
using the scale transformation, we obtain%
\begin{equation*}
r_{T_{m}}^{\eta -1/2}\max_{r_{T_{m}}\leq \tau \leq T_{m}}\tau ^{-\eta
}\left\vert W_{T_{m}}\left( \tau \right) \right\vert \overset{D}{=}\left(
\frac{r_{T_{m}}}{T_{m}}\right) ^{\eta -1/2}\max_{\frac{r_{T_{m}}}{T_{m}}\leq
t\leq 1}t^{-\eta }\left\vert W\left( t\right) \right\vert ;
\end{equation*}%
setting $s=t\frac{T_{m}}{r_{T_{m}}}$, and using again the scale
transformation%
\begin{equation*}
\left( \frac{r_{T_{m}}}{T_{m}}\right) ^{\eta -1/2}\max_{\frac{r_{T_{m}}}{%
T_{m}}\leq t\leq 1}t^{-\eta }\left\vert W\left( t\right) \right\vert \overset%
{D}{=}\max_{1\leq s\leq \frac{T_{m}}{r_{T_{m}}}}s^{-\eta }\left\vert W\left(
s\right) \right\vert \overset{a.s.}{\rightarrow }\sup_{1\leq s<\infty }\frac{%
\left\vert W\left( s\right) \right\vert }{s^{\eta }},
\end{equation*}%
completing the proof.
\end{proof}

\begin{proof}[Proof of Theorem \ref{function-power}]

We write
\begin{equation*}
S_{\tau }=\sum_{j=1}^{\tau }z_{j}+\sum_{j=1}^{t^{\ast }}\psi
_{j}+\sum_{j=t^{\ast }+1}^{\tau }\psi _{j},
\end{equation*}%
with the convention that $\sum_{c}^{d}=0$ whenever $d<c$, and we begin by
observing that, under (\ref{b1}), (\ref{b11}) entails that there are two
numbers $0<a<b<1$ such that there exists a triplet of random variables $%
\left( p_{1,0},p_{2,0},T_{0}\right) $ and a positive constant $c_{0}$ such
that, for $p_{1}\geq p_{1,0}$, $p_{2}\geq p_{2,0}$ and $T\geq T_{0}$%
\begin{equation}
\widehat{\lambda }_{k_{1}+1,\tau }\geq c_{0}p_{1}\text{ for }t^{\ast
}+am\leq \tau \leq t^{\ast }+m-bm.  \label{lamb11}
\end{equation}

We now turn to proving the theorem. Consider first the case $0\leq \eta <%
\frac{1}{2}$, and note that, by the Law of the Iterated Logarithm%
\begin{equation*}
T_{m}^{\eta -1/2}\max_{1\leq \tau \leq T_{m}}\tau ^{-\eta }\sum_{j=1}^{\tau
}z_{j}=O_{a.s.}\left( \sqrt{\ln \ln T_{m}}\right) ;
\end{equation*}%
also%
\begin{equation*}
T_{m}^{\eta -1/2}\max_{1\leq \tau \leq T_{m}}\tau ^{-\eta
}\sum_{j=1}^{t^{\ast }}\psi _{j}=o_{a.s.}\left( T_{m}^{\eta +1/2}g\left(
\frac{p_{1}^{1-\delta }}{m^{1/2}}\left( \ln p_{1}\right) ^{1+\epsilon
}\left( \ln p_{2}\right) ^{1+\epsilon }\left( \ln m\right) ^{1+\epsilon
}\right) \right) =o_{a.s.}\left( 1\right) ,
\end{equation*}%
on account of (\ref{restriction}). Finally, by (\ref{lamb11}), there exists
a triplet of random variables $\left( p_{1,0},p_{2,0},T_{0}\right) $ and a
positive constant $c_{0}$ such that, for $p_{1}\geq p_{1,0}$, $p_{2}\geq
p_{2,0}$ and $T\geq T_{0}$%
\begin{eqnarray}
&&\max_{1\leq \tau \leq T_{m}}\tau ^{-\eta }\sum_{j=t^{\ast }+1}^{\tau }\psi
_{j}  \label{max-eta} \\
&\geq &\max_{1\leq \tau \leq T_{m}}\tau ^{-\eta }\sum_{j=am+t^{\ast
}+1}^{\tau }\psi _{j}  \notag \\
&\geq &c_{0}g\left( p_{1}^{1-\delta }\right) \max_{t^{\ast }+1\leq \tau \leq
m+t^{\ast }}\tau ^{-\eta }\left( \tau -t^{\ast }\right)  \notag \\
&=&c_{1}g\left( p_{1}^{1-\delta }\right) \left( t^{\ast }\right) ^{1-\eta },
\notag
\end{eqnarray}%
where the last result follows from standard algebra, and it is valid for all
$\eta $. Now (\ref{power-1}) follows putting everything together. When $\eta
>1/2$, essentially the same arguments yield%
\begin{eqnarray*}
r_{T_{m}}^{\eta -1/2}\max_{r_{T_{m}}\leq \tau \leq T_{m}}\tau ^{-\eta
}\sum_{j=1}^{\tau }z_{j} &=&O_{a.s.}\left( \sqrt{\ln \ln T_{m}}\right) , \\
r_{T_{m}}^{\eta -1/2}\max_{r_{T_{m}}\leq \tau \leq T_{m}}\tau ^{-\eta
}\sum_{j=1}^{t^{\ast }}\psi _{j} &=&o_{a.s.}\left( r_{T_{m}}^{-1/2}t^{\ast
}g\left( \frac{p_{1}^{1-\delta }}{m^{1/2}}\left( \ln p_{1}\right)
^{1+\epsilon }\left( \ln p_{2}\right) ^{1+\epsilon }\left( \ln m\right)
^{1+\epsilon }\right) \right) =o_{a.s.}\left( 1\right) .
\end{eqnarray*}%
Equation (\ref{power-3}) follows by putting these results and (\ref{max-eta}%
) together. Finally, when $\eta =1/2$, it is easy to see by the same logic
as above that%
\begin{eqnarray*}
\max_{1\leq \tau \leq T_{m}}\tau ^{-1/2}\sum_{j=1}^{\tau }z_{j}
&=&O_{a.s.}\left( \sqrt{\ln \ln T_{m}}\right) , \\
\max_{1\leq \tau \leq T_{m}}\tau ^{-1/2}\sum_{j=1}^{t^{\ast }}\psi _{j}
&=&o_{a.s.}\left( t^{\ast }g\left( \frac{p_{1}^{1-\delta }}{m^{1/2}}\left(
\ln p_{1}\right) ^{1+\epsilon }\left( \ln p_{2}\right) ^{1+\epsilon }\left(
\ln m\right) ^{1+\epsilon }\right) \right) =o_{a.s.}\left( 1\right) ,
\end{eqnarray*}%
so that (\ref{power-2}) follows from the same logic as above. Under (\ref{b2}%
), the proof uses essentially the same arguments, with the only difference
that, by (\ref{b21}), there is a number $0<a<1$ such that there exists a
triplet of random variables $\left( p_{1,0},p_{2,0},T_{0}\right) $ and a
positive constant $c_{0}$ such that, for $p_{1}\geq p_{1,0}$, $p_{2}\geq
p_{2,0}$ and $T\geq T_{0}$%
\begin{equation*}
\widehat{\lambda }_{k_{1}+1,\tau }\geq c_{0}p_{1}\text{ for }t^{\ast
}+am\leq \tau \leq T_{m}.
\end{equation*}

\end{proof}

\begin{proof}[Proof of Theorem \ref{monitor-evt}]

We will use the short-hand notation%
\[
Z_{Tm}=\max_{1\leq \tau \leq T_{m}}y_{\tau }.
\]%
Note that%
\[
P^{\ast }\left( \frac{Z_{Tm}-b_{Tm}}{a_{Tm}}\leq v\right) =P^{\ast }\left(
Z_{Tm}\leq a_{Tm}v+b_{Tm}\right) ,
\]%
and consider the sequence $y_{\tau }$\ defined in (\ref{yt}), viz.
\[
y_{\tau }=z_{\tau }+\psi _{\tau }.
\]%
By construction, $z_{\tau }$ is independent across $\tau $ and independent
of the sample $\left\{ X_{t},1\leq t\leq T+T_{m}\right\} $. Hence, letting $%
\Phi \left( \cdot \right) $ be the distribution of the standard normal and $%
\varphi \left( \cdot \right) $ its density, we can write%
\begin{eqnarray}
P^{\ast }\left( Z_{Tm}\leq a_{Tm}v+b_{Tm}\right)  &=&\prod\limits_{\tau
=1}^{T_{m}}P^{\ast }\left( y_{\tau }\leq a_{Tm}v+b_{Tm}\right)
\label{max-1} \\
&=&\prod\limits_{\tau =1}^{T_{m}}P^{\ast }\left( z_{\tau }\leq
a_{Tm}v+b_{Tm}-\psi _{\tau }\right)   \nonumber \\
&=&\exp \left( \sum_{\tau =1}^{T_{m}}\ln \Phi \left( a_{Tm}v+b_{Tm}-\psi
_{\tau }\right) \right) .  \nonumber
\end{eqnarray}%
Note that%
\begin{equation}
\ln \Phi \left( a_{Tm}v+b_{Tm}-\psi _{\tau }\right) =\ln \Phi \left(
a_{Tm}v+b_{Tm}\right) +\ln \frac{\Phi \left( a_{Tm}v+b_{Tm}-\psi _{\tau
}\right) }{\Phi \left( a_{Tm}v+b_{Tm}\right) }.  \label{max-2}
\end{equation}%
Using Taylor's expansion, for some $c_{0}>0$ it holds that%
\begin{eqnarray}
\ln \Phi \left( a_{Tm}v+b_{Tm}-\psi _{\tau }\right)  &=&\ln \Phi \left(
a_{Tm}v+b_{Tm}\right) +\ln \frac{\Phi \left( a_{Tm}v+b_{Tm}-\psi _{\tau
}\right) }{\Phi \left( a_{Tm}v+b_{Tm}\right) }  \label{max-3} \\
&=&\ln \Phi \left( a_{Tm}v+b_{Tm}\right) +\ln \left[ 1-c_{0}\psi _{\tau }%
\frac{\varphi \left( a_{Tm}v+b_{Tm}\right) }{\Phi \left(
a_{Tm}v+b_{Tm}\right) }\right] .  \nonumber
\end{eqnarray}%
Therefore, putting together (\ref{max-1})-(\ref{max-3})%
\begin{eqnarray}
&&P^{\ast }\left( Z_{Tm}\leq a_{Tm}v+b_{Tm}\right)   \label{max-4} \\
&=&\exp \left( \sum_{\tau =1}^{T_{m}}\ln \Phi \left( a_{Tm}v+b_{Tm}-\psi
_{\tau }\right) \right)   \nonumber \\
&=&\exp \left( \sum_{\tau =1}^{T_{m}}\ln \Phi \left( a_{Tm}v+b_{Tm}\right)
\right) \exp \left( \sum_{\tau =1}^{T_{m}}\ln \left[ 1-c_{0}\psi _{\tau }%
\frac{\varphi \left( a_{Tm}v+b_{Tm}\right) }{\Phi \left(
a_{Tm}v+b_{Tm}\right) }\right] \right)   \nonumber \\
&=&\left[ \Phi \left( a_{Tm}v+b_{Tm}\right) \right] ^{T_{m}}\exp \left(
\sum_{\tau =1}^{T_{m}}\ln \left[ 1-c_{1}\psi _{\tau }\right] \right) .
\nonumber
\end{eqnarray}%
Consider the elementary inequality $\ln \left( 1-c_{1}\psi _{\tau }\right)
\geq -c_{2}\psi _{\tau }$, for some $c_{2}>0$; thus, from (\ref{max-4})%
\[
1\geq \exp \left( \sum_{\tau =1}^{T_{m}}\ln \left[ 1-c_{1}\psi _{\tau }%
\right] \right) \geq \exp \left( \sum_{\tau =1}^{T_{m}}-c_{2}\psi _{\tau
}\right) .
\]%
By using (\ref{restriction}) it follows that $\lim_{p_{1},p_{2},m\rightarrow
\infty }T_{m}\psi _{\tau }=0$, so that the final result obtains by dominated
convergence.

Under both alternatives (\ref{b1}) and (\ref{b2}), it holds that there is a
constant $c_{0}$ and a random variable $m_{0}$ such that, for $m\geq m_{0}$,
we have $\widehat{\lambda }_{k_{1}+1,\widetilde{\tau }}\geq c_{0}p_{1}$, for
at least one $t^{\ast }+1\leq \widetilde{\tau }\leq T_{m}$. Thus we have
\begin{eqnarray*}
P^{\ast }\left( Z_{Tm}\leq a_{Tm}v+b_{Tm}\right) &=&\prod\limits_{\tau
=1}^{T_{m}}P^{\ast }\left( y_{\tau }\leq a_{Tm}v+b_{Tm}\right) \\
&=&P^{\ast }\left( y_{\widetilde{\tau }}\leq a_{Tm}v+b_{Tm}\right)
\prod\limits_{\tau \neq \widetilde{\tau }}P^{\ast }\left( y_{\tau }\leq
a_{Tm}v+b_{Tm}\right) \\
&\leq &P^{\ast }\left( y_{\widetilde{\tau }}\leq a_{Tm}v+b_{Tm}\right) .
\end{eqnarray*}%
Note that%
\[
P^{\ast }\left( y_{\widetilde{\tau }}\leq a_{Tm}v+b_{Tm}\right) =P^{\ast
}\left( z_{\widetilde{\tau }}\leq a_{Tm}v+b_{Tm}-\psi _{\widetilde{\tau }%
}\right) =\Phi \left( a_{Tm}v+b_{Tm}-\psi _{\widetilde{\tau }}\right)
\]%
and, on account of (\ref{alt-max})
\[
P^{\ast }\left( \omega :\lim_{p_{1},p_{2},m\rightarrow \infty
}a_{Tm}v+b_{Tm}-\psi _{\widetilde{\tau }}=-\infty \right) =1,
\]%
for all $-\infty <v<\infty $, because $a_{Tm}$ and $b_{Tm}$ are both $%
O\left( \sqrt{\ln T_{m}}\right) $. Using equation (5) in \citet{borjesson}%
\[
\Phi \left( a_{Tm}v+b_{Tm}-\psi _{\widetilde{\tau }}\right) \leq \frac{\exp
\left( -\frac{1}{2}\left( a_{Tm}v+b_{Tm}-\psi _{\widetilde{\tau }}\right)
^{2}\right) }{\sqrt{2\pi }\left\vert a_{Tm}v+b_{Tm}-\psi _{\widetilde{\tau }%
}\right\vert },
\]%
and therefore it follows that, as $\min \left( p_{1},p_{2},m\right)
\rightarrow \infty $, $P^{\ast }\left( z_{\tau }\leq a_{Tm}v+b_{Tm}-\psi _{%
\widetilde{\tau }}\right) =0$ a.s. conditionally on the sample. Thus, by
dominated convergence, as $\min \left( p_{1},p_{2},m\right) \rightarrow
\infty $ it holds that%
\[
\ P^{\ast }\left( Z_{\widetilde{\tau }}\leq a_{Tm}v+b_{Tm}\right) =0,\ \
a.s.
\]
conditionally on the sample. This implies the desired result.

\end{proof}

\renewcommand{\thetable}{A.\arabic{table}}

\subsection{Further simulations and sensitivity analysis\label{sensitivity}}

\subsubsection{Empirical size and power in the presence of changes in $C$}

When conducting monitoring on the loading matrix $R$, it cannot be
guaranteed that the other loading matrix $C$ is always invariant. Hence, in
this subsection, we investigate the effects on our monitoring procedure if
the factor loading $C$ also changes. Under the null of no change in $R$, we
generate data by
\[
X_{t}=\left\{ \begin{matrix} & RF_{t}C^{\prime }+E_{t}, & 1\leq t\leq \tau ,
\\ & RF_{t}C_{new}^{\prime }+E_{t}, & \tau +1\leq t\leq T,\end{matrix}\right.
\]%
where $C_{new}$ is regenerated after time point $\tau +1$ with \textit{i.i.d.%
} entries from $\mathcal{U}(-\sqrt{3},\sqrt{3})$. That is, only the loading
matrix $C$ changes. All the other parameters are set to the same values as
those introduced in Section \ref{simulation}. The empirical sizes over $%
1,000 $ replications are reported in Table \ref{tab5}. Compared with the
results reported in Table \ref{tab1}, we conclude that the monitoring
procedures have controlled empirical sizes no matter whether the loading
matrix $C$ changes or not, i.e, the change of the loading matrix $C$ have
limited impact on the empirical sizes when we detect changes for loading
matrix $R$.

\begin{table*}[tbph]
\caption{Empirical powers under \protect\ref{b1} (loading space changes)
over 1000 replications. $T=200$, $k_{1}=k_{2}=3$, $\protect\epsilon =0.05$.}
\begin{center}
{\small \ \addtolength{\tabcolsep}{0pt} \renewcommand{\arraystretch}{1.2}
\scalebox{0.75}{ 	
			\begin{tabular*}{21cm}{ccccccccccccccc}
				\toprule[1.2pt]
				&&&\multicolumn{6}{l}{$\alpha=0.05$}	&\multicolumn{6}{l}{$\alpha=0.10$}\\\cmidrule(lr){4-9}\cmidrule(lr){10-15}
				&&&\multicolumn{5}{l}{Partial-sum}&Worst&\multicolumn{5}{l}{Partial-sum}&Worst\\\cmidrule(lr){4-8}\cmidrule(lr){10-14}
				$m$&$p_1$&$p_2$&$\eta=0$&$\eta=0.25$&$\eta=0.5$&$\eta=0.65$&$\eta=0.75$	&case&$\eta=0$&$\eta=0.25$&$\eta=0.5$&$\eta=0.65$&$\eta=0.75$	&case	\\\midrule[1.2pt]
50&50&20&1&1&1&1&1&1&1&1&1&1&1&1
\\
50&50&50&1&1&1&1&1&1&1&1&1&1&1&1
\\
50&50&80&1&1&1&1&1&1&1&1&1&1&1&1
\\
50&80&20&1&1&1&1&1&1&1&1&1&1&1&1
\\
50&80&50&1&1&1&1&1&1&1&1&1&1&1&1
\\
50&80&80&1&1&1&1&1&1&1&1&1&1&1&1
\\
50&100&20&1&1&1&1&1&1&1&1&1&1&1&1
\\
50&100&50&1&1&1&1&1&1&1&1&1&1&1&1
\\
50&100&80&1&1&1&1&1&1&1&1&1&1&1&1
\\\midrule[1.2pt]
80&50&20&1&1&1&1&1&1&1&1&1&1&1&1
\\
80&50&50&1&1&1&1&1&1&1&1&1&1&1&1
\\
80&50&80&1&1&1&1&1&1&1&1&1&1&1&1
\\
80&80&20&1&1&1&1&1&1&1&1&1&1&1&1
\\
80&80&50&1&1&1&1&1&1&1&1&1&1&1&1
\\
80&80&80&1&1&1&1&1&1&1&1&1&1&1&1
\\
80&100&20&1&1&1&1&1&1&1&1&1&1&1&1
\\
80&100&50&1&1&1&1&1&1&1&1&1&1&1&1
\\
80&100&80&1&1&1&1&1&1&1&1&1&1&1&1
\\\midrule[1.2pt]
100&50&20&1&1&1&1&1&1&1&1&1&1&1&1
\\
100&50&50&1&1&1&1&1&1&1&1&1&1&1&1
\\
100&50&80&1&1&1&1&1&1&1&1&1&1&1&1
\\
100&80&20&1&1&1&1&1&1&1&1&1&1&1&1
\\
100&80&50&1&1&1&1&1&1&1&1&1&1&1&1
\\
100&80&80&1&1&1&1&1&1&1&1&1&1&1&1
\\
100&100&20&1&1&1&1&1&1&1&1&1&1&1&1
\\
100&100&50&1&1&1&1&1&1&1&1&1&1&1&1
\\
100&100&80&1&1&1&1&1&1&1&1&1&1&1&1\\
				\bottomrule[1.2pt]		
		\end{tabular*}} }
\end{center}
\end{table*}

\begin{table*}[tbph]
\caption{Empirical sizes ($\%$) under $\tilde{H}_{0}$ (only $C$ changes)
over 1000 replications. $T=200$, $k_{1}=k_{2}=3$, $\protect\epsilon =0.05$.}
\label{tab5}
\begin{center}
{\small \ \addtolength{\tabcolsep}{0pt} \renewcommand{\arraystretch}{1.2}
\scalebox{0.75}{ 	
			\begin{tabular*}{21cm}{ccccccccccccccc}
				\toprule[1.2pt]
				&&&\multicolumn{6}{l}{$\alpha=0.05$}	&\multicolumn{6}{l}{$\alpha=0.10$}\\\cmidrule(lr){4-9}\cmidrule(lr){10-15}
				&&&\multicolumn{5}{l}{Partial-sum}&Worst&\multicolumn{5}{l}{Partial-sum}&Worst\\\cmidrule(lr){4-8}\cmidrule(lr){10-14}
				$m$&$p_1$&$p_2$&$\eta=0$&$\eta=0.25$&$\eta=0.5$&$\eta=0.65$&$\eta=0.75$	&case&$\eta=0$&$\eta=0.25$&$\eta=0.5$&$\eta=0.65$&$\eta=0.75$	&case	\\\midrule[1.2pt]
50&50&20&5.5&5.6&3.1&5.2&4.2&3.8&10.7&10.5&8.1&7.9&7.7&9.4
\\
50&50&50&3.8&4&1.6&2.2&2.1&4.6&7.9&8.1&4.7&5.9&6.5&10.5
\\
50&50&80&4.6&4.6&1.9&3.7&3.4&4.2&8.4&8.8&6.5&6.9&6.3&9.8
\\
50&80&20&3.8&4&1.4&2.2&2.4&4.6&9.6&9.7&4.8&6.5&7.1&8.4
\\
50&80&50&4.5&4.7&1.9&3.3&3.1&4.8&9&9.1&5.5&7.2&6.6&8.8
\\
50&80&80&4.9&4.9&1.2&2.9&3.6&5.1&10.3&10&5.4&6&6&10.6
\\
50&100&20&5.2&5.1&1.8&3.3&3.3&4.4&10.8&10.7&7&7&6.3&9.2
\\
50&100&50&4.4&4.5&1.1&2.8&2.7&3&10.8&9.9&5.2&7.5&8.1&7.9
\\
50&100&80&3.9&4.3&1.7&3.1&3&5&8&7.2&5.6&7.5&7.5&9.2
\\\midrule[1.2pt]
80&50&20&4.5&3.7&1.7&3.7&3&3&8.9&8.8&5.7&6.8&6.7&7.6
\\
80&50&50&4.3&4.6&1.8&2.6&2.5&4.5&9.3&8.5&5.5&5.9&5.7&9.8\\
80&50&80&5&4.7&2.1&2.9&2.8&4.1&10&8.8&6.7&6.9&6&9.1
\\
80&80&20&3.8&3.4&1.3&3&3.6&3.5&8.6&8.2&6&8&7.3&7.9
\\
80&80&50&4&4.1&1.7&2.6&2.3&3.8&9.8&8.9&5.7&7.3&6.6&9.9\\
80&80&80&5.3&4.6&1.8&2.3&2.6&4.3&9.6&9.2&5.9&6.3&6.2&9.2\\
80&100&20&4.2&4.3&1.7&2.8&2.6&4.5&9&8.2&5.3&5.8&6.3&9
\\
80&100&50&3.9&4.1&1.7&3.2&3.4&4.2&8.4&7.9&5.5&6.5&6.8&7.7\\
80&100&80&3.9&3.3&1.1&2.3&2&3.7&7.9&9.3&4.1&6.4&6.8&9.5
\\\midrule[1.2pt]
100&50&20&4.1&3.9&1.5&2.4&2.3&4.7&8.5&8.3&5.3&5.9&5.8&10.2\\
100&50&50&5.1&4.4&1.4&3.1&3.1&5&11.6&10.1&5.7&6.7&6.1&9.5
\\
100&50&80&4.9&4.2&2&1.7&1.6&4.2&8.9&8.6&5&5&5.5&8.8
\\
100&80&20&5&4.5&2&4.1&3.6&4.9&10.2&9.8&7.9&7.5&6.9&10.5\\
100&80&50&5.5&4.7&1.6&3.6&3.5&3.6&10.4&10.8&6.1&8.3&7.9&9.1\\
100&80&80&4&3.7&1.7&2.3&2.2&2.9&8.4&7.6&5.2&7.2&6.8&7.5
\\
100&100&20&4.9&4.3&2&3&2.7&3.5&9.3&8.6&5.6&6.4&5.8&7.3
\\
100&100&50&5&4.7&1.8&3.3&3.4&4.7&10.4&9.3&5.9&6.7&7.2&9.3\\
100&100&80&3.4&3.2&1.5&3.5&3.3&3.2&8&7.6&5.8&6.7&6.8&7\\
				\bottomrule[1.2pt]		
		\end{tabular*}} }
\end{center}
\end{table*}

Under the alternative hypothesis (say $H_{A3}$), data are generated as
\[
X_{t}=\left\{ \begin{matrix} & RF_{t}C^{\prime }+E_{t}, & 1\leq t\leq \tau ,
\\ & R_{new}F_{t}C_{new}^{\prime }+E_{t}, & \tau +1\leq t\leq T,\end{matrix}%
\right.
\]%
where $R_{new}$ and $C_{new}$ are both separately regenerated from the time
point $\tau +1$ with i.i.d. entries from $\mathcal{U}(-\sqrt{3},\sqrt{3})$.
We let $R$ and $C$ change at the same time only for simplicity of
data-generating. Basically they can change at different time and this has
little effects on the results. All the other parameters are set the same as
in Table \ref{tab3}. Under $H_{A3}$, our simulation results show that all
the procedures have empirical power equal to 1, and the median delays of
detection are reported in Table \ref{tab6}. Unsurprisingly, compared with
Table \ref{tab3}, the change of $C$ has limited effects on the procedures.

\begin{table*}[tbph]
\caption{Median delays for detecting change point of $R$ under $H_{A3}$ ($C$
also changes) over 1000 replications. $T=200$, $k_{1}=k_{2}=3$, $\protect%
\epsilon =0.05$.}
\label{tab6}
\begin{center}
{\small \ \addtolength{\tabcolsep}{0pt} \renewcommand{\arraystretch}{1.2}
\scalebox{0.75}{ 	
			\begin{tabular*}{21cm}{ccccccccccccccc}
				\toprule[1.2pt]
				&&&\multicolumn{6}{l}{$\alpha=0.05$}	&\multicolumn{6}{l}{$\alpha=0.10$}\\\cmidrule(lr){4-9}\cmidrule(lr){10-15}
				&&&\multicolumn{5}{l}{Partial-sum}&Worst&\multicolumn{5}{l}{Partial-sum}&Worst\\\cmidrule(lr){4-8}\cmidrule(lr){10-14}
				$m$&$p_1$&$p_2$&$\eta=0$&$\eta=0.25$&$\eta=0.5$&$\eta=0.65$&$\eta=0.75$	&case&$\eta=0$&$\eta=0.25$&$\eta=0.5$&$\eta=0.65$&$\eta=0.75$	&case	\\\midrule[1.2pt]
50&50&20&5&5&5&5&5&4&5&5&5&5&5&3
\\
50&50&50&3&3&3&3&3&2&3&3&3&3&3&2
\\
50&50&80&3&3&3&3&3&2&3&3&3&3&3&2
\\
50&80&20&5&5&5&5&5&4&5&5&5&5&5&4
\\
50&80&50&3&3&3&3&3&2&3&3&3&3&3&2
\\
50&80&80&3&2&2&2&3&2&2&2&2&2&2&2
\\
50&100&20&5&5&5&5&5&4&5&5&5&5&5&3
\\
50&100&50&3&3&3&3&3&2&3&3&3&3&3&2
\\
50&100&80&3&2&2&2&3&2&2&2&2&2&2&2
\\\midrule[1.2pt]
80&50&20&6&6&5&5&6&5&6&5&5&5&5&4
\\
80&50&50&5&4&4&4&4&4&5&4&4&4&4&4
\\
80&50&80&5&4&4&4&4&4&5&4&4&4&4&3
\\
80&80&20&6&6&6&6&6&5&6&6&5&5&6&4
\\
80&80&50&4&4&3.5&3&4&3&4&3&3&3&3&3
\\
80&80&80&3&3&3&3&3&2&3&3&3&3&3&2
\\
80&100&20&6&6&6&6&6&5&6&6&5&5&5.5&5\\
80&100&50&4&4&4&4&4&3&4&4&3&3&3&3
\\
80&100&80&3&3&3&3&3&2&3&3&3&3&3&2
\\\midrule[1.2pt]
100&50&20&7&6&5&5&5&5&6&6&5&5&5&5
\\
100&50&50&6&5&5&5&5&5&6&5&5&4&4&5
\\
100&50&80&6&5&5&4&4&4&6&5&5&4&4&4
\\
100&80&20&7&6&6&5&5&5&7&6&5&5&5&5
\\
100&80&50&4&4&3&4&4&3&4&4&3&4&4&3
\\
100&80&80&4&3&3&4&4&3&4&3&3&4&4&3
\\
100&100&20&7&6&6&5&5&5&7&6&5&5&5&5
\\
100&100&50&4&4&3&4&4&3&4&4&3&4&4&3
\\
100&100&80&3&3&3&4&4&3&3&3&3&4&4&3\\
				\bottomrule[1.2pt]		
		\end{tabular*}} }
\end{center}
\end{table*}

\subsubsection{Sensitivity to $T$ and $k_{2}$}

\begin{table*}[tbph]
\caption{Sensitivity to $T$ and $k_{2}$ over 1000 replications. $m=p_{1}=100$%
, $p_{2}=80$, $k_{1}=3$, $\protect\epsilon =0.05$.}
\label{tab8}
\begin{center}
{\small \ \addtolength{\tabcolsep}{0pt} \renewcommand{\arraystretch}{1.2}
\scalebox{0.75}{ 	
			\begin{tabular*}{21cm}{ccccccccccccccc}
				\toprule[1.2pt]
				&&\multicolumn{6}{l}{$\alpha=0.05$}	&\multicolumn{6}{l}{$\alpha=0.10$}\\\cmidrule(lr){3-8}\cmidrule(lr){9-14}
				&&\multicolumn{5}{l}{Partial-sum}&Worst&\multicolumn{5}{l}{Partial-sum}&Worst\\\cmidrule(lr){3-7}\cmidrule(lr){9-13}
				$T$&$k_2$&$\eta=0$&$\eta=0.25$&$\eta=0.5$&$\eta=0.65$&$\eta=0.75$	&case&$\eta=0$&$\eta=0.25$&$\eta=0.5$&$\eta=0.65$&$\eta=0.75$	&case	\\\midrule[1.2pt]
				\multicolumn{14}{l}{Empirical sizes ($\%$)}\\\midrule[1.2pt]
200&1&4.4&4.3&2.3&4.4&3.3&4&9.2&8.6&6.2&8&8.1&9.2
\\
200&3&5.2&4.8&2.3&3.1&3.1&5.8&9.4&10.4&6.3&6.3&6.5&11.6\\
200&5&3.8&3.4&1.7&2.6&3.4&4.5&8.1&7.9&5.6&7&6.9&9.3
\\
300&1&4&4.4&1.8&3&2.9&4.3&9.4&8&5&5.7&6&9
\\
300&3&4.3&4.2&2&2.7&2.4&4.4&9.2&9.2&6.1&6.9&6&10.6\\
300&5&5.1&4.7&1.6&3.4&3.6&4&9&8.5&6.3&6.6&7&8.2
\\
400&1&5.7&5&2&2.4&2.4&4.8&9.8&9.7&6.3&7&7.1&10.3
\\
400&3&3.9&3.7&1&2.7&2.5&5.8&8.6&7.8&4.6&5.8&5.9&10.3\\
400&5&3.3&3.5&2&3.3&2.9&5.2&7.1&7.2&5.8&7.4&7.5&12.2\\\midrule[1.2pt]
				\multicolumn{14}{l}{Median delays}\\\midrule[1.2pt]
200&1&3&2&2&4&4&2&3&2&2&4&4&2
\\
200&3&3&3&3&4&4&3&3&3&3&4&4&3
\\
200&5&4&3&3&4&4&3&4&3&3&4&4&3
\\
300&1&3&3&3&3&3&2&3&3&3&3&3&2
\\
300&3&4&3&3&3&4&3&4&3&3&3&3&3
\\
300&5&4&4&4&4&4&3&4&4&4&4&4&3
\\
400&1&3&3&3&3&3&2&3&3&3&3&3&2
\\
400&3&4&4&4&4&4&3&4&3&3&3&4&3
\\
400&5&4&4&4&4&4&3&4&4&4&4&4&3\\
				\bottomrule[1.2pt]		
		\end{tabular*}} }
\end{center}
\end{table*}

\begin{table*}[!h]
\caption{Sensitivity to $\protect\delta $ over 1000 replications. $T=200$, $%
m=p_{1}=100$, $p_{2}=80$, $k_{1}=k_{2}=3$.}
\label{tab9}
\begin{center}
{\small \ \addtolength{\tabcolsep}{0pt} \renewcommand{\arraystretch}{1.2}
\scalebox{0.75}{ 	
			\begin{tabular*}{21cm}{ccccccccccccccc}
				\toprule[1.2pt]
				&\multicolumn{6}{l}{$\alpha=0.05$}	&\multicolumn{6}{l}{$\alpha=0.10$}\\\cmidrule(lr){2-7}\cmidrule(lr){8-13}
				&\multicolumn{5}{l}{Partial-sum}&Worst&\multicolumn{5}{l}{Partial-sum}&Worst\\\cmidrule(lr){2-6}\cmidrule(lr){8-12}
				$\epsilon$&$\eta=0$&$\eta=0.25$&$\eta=0.5$&$\eta=0.65$&$\eta=0.75$	&case&$\eta=0$&$\eta=0.25$&$\eta=0.5$&$\eta=0.65$&$\eta=0.75$	&case	\\\midrule[1.2pt]
				\multicolumn{13}{l}{Empirical sizes ($\%$)}\\\midrule[1.2pt]
0.02&4.4&4&0.7&2.4&2.4&3.8&7.7&8.1&4.8&6&6&8.2
\\
0.03&5.2&4.8&2.3&3.1&3.1&5.8&9.4&10.4&6.3&6.3&6.5&11.6\\
0.04&4.9&4.9&1.6&2.6&2.9&3.1&9.6&9.4&5.5&5.7&5.8&8.1
\\
0.05&3.9&3.8&1.3&3.1&2.9&4.7&8.1&8.1&5.5&5.3&5.9&10.1
\\
0.06&3.8&3.3&1.5&3.9&4.2&4.2&8.5&7.8&5.3&8&8.3&8.9\\\midrule[1.2pt]
				\multicolumn{13}{l}{Median delays}\\\midrule[1.2pt]
0.02&3&3&2&4&4&2&3&2&2&4&4&2
\\
0.03&3&3&2&4&4&2&3&3&2&4&4&2
\\
0.04&3&3&3&4&4&2&3&3&2&4&4&2
\\
0.05&3&3&3&4&4&3&3&3&3&4&4&3
\\
0.06&4&3&3&4&4&3&4&3&3&4&4&3\\
				\bottomrule[1.2pt]		
		\end{tabular*}} }
\end{center}
\end{table*}

In all the above simulations, we have used $T=200$ and $k_{2}=3$. In this
subsection, we consider other choices for $T$ and $k_{2}$, but fix $%
m=p_{1}=100$ and $p_{2}=80$. All the other settings are the same as those in
Section \ref{simulation}. The empirical sizes and median delays are reported
in Table \ref{tab8}, while all the empirical powers under (\ref{b1}) are 1.
The results show that $T$ and $k_{2}$ may have slight effects on the sizes,
partially due to the small dimensions and slow convergence rates. On the
other hand, the median delays are not sensitive to $T$ and $k_{2}$.

\subsubsection{Sensitivity to $\protect\delta $}

In a last set of experiments, we investigate the effects of $\delta $ when
calculating $\psi _{\tau }$ in (\ref{psi}). Note that $\delta $ is mainly
determined by the constant $\epsilon $ in (\ref{equ:deltabeta}), as well as
the relative sample sizes. We consider $\epsilon \in \{0.02,0.03,\ldots
,0.06\}$, and use $m=p_{1}=100$, $p_{2}=80$, while all the other settings
are the same as those introduced in Section \ref{simulation}. The empirical
sizes and median delays are reported in Table \ref{tab9} over $1,000$
replications. It can be concluded from the table that $\epsilon $ has
limited effects on the results.

\end{document}